%% file: neurips_2026.tex
\documentclass{article}

\PassOptionsToPackage{numbers,compress}{natbib}

\usepackage[main, final]{neurips_2026}

\usepackage[utf8]{inputenc} 
\usepackage[T1]{fontenc}    
\usepackage{hyperref}       
\usepackage{url}            
\usepackage{graphicx}       
\usepackage{booktabs}       
\usepackage{amsfonts}       
\usepackage{nicefrac}       
\usepackage{microtype}      
\usepackage{xcolor}         
\usepackage{placeins} 
\usepackage{float}          
\usepackage{algorithm}
\usepackage{algpseudocode}
\usepackage{amsfonts,amsthm,amsmath} 

\theoremstyle{definition}
\newtheorem{theorem}{Theorem}[section]
\newtheorem{lemma}{Lemma}[section]
\newtheorem{definition}{Definition}[section]
\newtheorem{example}{Example}[section]
\newtheorem{remark}{Remark}[section]
\newtheorem{corollary}{Corollary}[section]
\newtheorem{assumption}{Assumption}[section]
\newtheorem{proposition}{Proposition}[section]

\input{Math_command}
\title{Impossibility of Distribution-Free Predictive Inference for Individual Treatment Effects}

\author{%
Chongguang Tao\thanks{These authors are co-first authors, contributed equally to this work, and are listed in alphabetical order.} \quad
Zheng Zhou\footnotemark[1] \quad
Yuhong Yang\thanks{Corresponding author} \\
Qiuzhen College and Yau Mathematical Sciences Center, Tsinghua University \\
\texttt{\{taocg21, zhouz23\}@mails.tsinghua.edu.cn, yyangsc@tsinghua.edu.cn}
}

\begin{document}

\maketitle

\begin{abstract}
  Uncertainty quantification for individual treatment effects (ITEs) is a daunting challenge in causal inference. Motivated by recent advances in conformal prediction, several works aim to construct distribution-free prediction sets for ITEs with desired coverage under standard assumptions such as strong ignorability and overlap. In this paper, we show that such goals are fundamentally unattainable. Specifically, we establish finite-sample and asymptotic impossibility results demonstrating that any distribution-free prediction set achieving desired coverage for ITEs must be trivial, in the sense that it has infinite expected length. This infinite-expected-length conclusion requires no regularity assumption on the covariates and holds even in their absence. Our analysis builds on ideas from the hardness of conditional independence testing and exploits the missing-data nature of causal inference. These results provide a new perspective on existing methods, clarifying that their apparent success necessarily relies on additional structural assumptions beyond standard causal assumptions. \looseness=-1
\end{abstract}

\section{Introduction}

\input{Sections/Introduction}

\section{Main Results}

\label{sec:main_results}

\input{Sections/Notations}

\input{Sections/Main_results}

\section{Reviewing Some Previous Methods}
\input{Sections/Other_methods}

\section{Experiments}
\label{sec:experiments}
\input{Sections/Experiments.tex}

\section{Conclusions}
\input{Sections/Conclusions.tex}

\section*{Acknowledgements and Disclosures}
The authors thank the Yau Mathematical Sciences Center and Qiuzhen College at Tsinghua University for supporting this research and their participation in the conference. The authors declare no competing interests.

\newpage
\bibliographystyle{unsrtnat}
\bibliography{references}

\newpage
\appendix

\section{Proof of Main Theorems} 

\input{Sections/Proofs}

\section{Definitions of Stochastic Orders}

\input{Sections/Stochastic_orders_appendix}

\section{Detailed Experimental Setup and Numerical Results}

\input{Sections/Detailed_Exp}

\newpage
\input{checklist.tex}

\end{document}

%% file: Math_command.tex
\def\eqref#1{equation~\ref{#1}}

\def\1{\bm{1}}

\def\gB{{\mathcal{B}}}

\def\gO{{\mathcal{O}}}
\def\gP{{\mathcal{P}}}

\def\sR{{\mathbb{R}}}

\newcommand{\E}{\mathbb{E}}



%% file: Sections/Introduction.tex
Quantifying uncertainty for individual treatment effects (ITEs) is a daunting challenge in modern causal inference, with applications ranging from personalized medicine to policy evaluation \citep{kent2018personalized,kitagawa2018should,athey2021policy}. Although substantial progress has been made in estimating conditional average treatment effects (CATEs) \citep{hill2011bayesian,wager2018estimation,kunzel2019metalearners,nie2021quasi,kennedy2023towards}, uncertainty quantification for ITEs remains considerably more difficult and is especially appealing in high-stakes settings, where decisions are made for individuals rather than populations. A natural goal, therefore, is to construct prediction sets for ITEs with valid coverage.

Conformal prediction provides a general framework for constructing prediction sets with finite-sample, distribution-free coverage guarantees \citep{shafer2008tutorial,lei2018distribution,10.1561/2200000101,angelopoulos2026theoreticalfoundationsconformalprediction}, and has recently been extended to causal inference \citep{lei2021conformal,jin2023sensitivity,alaa2023conformal,wang2025conformal}. In particular, conformal methods for counterfactual inference \citep{lei2021conformal} and conformal meta-learners \citep{alaa2023conformal} suggest that nontrivial predictive inference for ITEs may be achievable in some settings. This naturally raises an essential question: under standard causal assumptions such as strong ignorability and overlap in treatment assignment \citep{rosenbaum1983central}, is distribution-free predictive inference for ITEs actually possible? \looseness=-1

In this paper, we answer this question in the negative. Under strong ignorability and standard overlap conditions, we establish finite-sample and asymptotic impossibility results showing that distribution-free prediction sets with valid or asymptotically valid coverage for the ITE must be uninformative. With continuous covariates, such sets must cover each pre-specified real value with high probability, yielding infinite expected length in finite samples and diverging expected length asymptotically. We further show that the infinite-expected-length conclusion requires no regularity assumption on the covariates and holds even in their absence (Proposition~\ref{prop:no_covariate_length}). These results show that informative distribution-free predictive inference for ITEs is infeasible without additional structural assumptions.

At a technical level, our analysis builds on ideas from \textbf{the hardness of conditional independence testing} \citep{10.1214/19-AOS1857}. For continuous covariates, we connect valid predictive inference for ITEs to conditional independence testing, a problem that is fundamentally intractable in a distribution-free sense. We further develop a complementary argument that removes the covariate continuity requirement for the infinite-expected-length conclusion. The result also reflects two difficulties specific to causal inference. First, ITE inference inherits the difficulty of \textbf{counterfactual prediction under covariate shift}: calibration across treatment groups depends on the propensity score $e(x)$, whose unknown nature makes reliable reweighting intrinsically difficult \citep{lei2021conformal,yang2024doubly}. This hardness propagates directly to methods such as \citep{lei2021conformal}, which construct ITE prediction sets by first predicting counterfactual outcomes. Second, the \textbf{missing data nature of causal inference} \citep{holland1986statistics} creates a deeper non-identifiability: for each unit, one potential outcome is unobserved, so different data-generating processes can induce the same observed distribution while implying arbitrarily different ITEs. Our proof exploits this indistinguishability, forcing any uniformly valid prediction set to be uninformative. \looseness=-1

Our analysis provides a new perspective on existing methods for ITE inference. The impossibility result implies that any procedure producing nontrivial prediction sets must rely on assumptions beyond strong ignorability and overlap. This clarifies the role of additional conditions in prior work: conformal inference for counterfactuals \citep{lei2021conformal} relies on \textbf{knowledge or accurate estimation of the propensity score} to validly reweight across treatment groups, whereas conformal meta-learners \citep{alaa2023conformal} depend on \textbf{stochastic dominance conditions} that cannot hold uniformly over the model class. Thus, the apparent success of existing methods reflects additional structure rather than a resolution of the fundamental hardness of ITE inference.

\paragraph{Our contributions.}
First, we establish finite-sample and asymptotic impossibility results for distribution-free predictive inference of individual treatment effects under strong ignorability and standard overlap, showing that any uniformly valid procedure must be trivial and have infinite expected length. The infinite-expected-length conclusion requires no regularity assumption on the covariates and holds even in their absence. Second, we explain this impossibility through two fundamental features of the problem, namely the hardness of conditional independence testing and the missing-data nature of causal inference, and also give a positive result under additional structural assumptions, showing that informative ITE prediction is possible under sensible additional conditions in the standard causal framework. Third, we revisit representative conformal methods for ITE inference, showing that their apparent success actually relies on additional assumptions, and use simulations to illustrate their limitations in practice.

%% file: Sections/Notations.tex
\subsection{Notations and Setup}

We work in the potential outcomes framework with binary treatment \citep{holland1986statistics}. For units \(i=1,\dots,n\), let \(T_i\in\{0,1\}\) denote the treatment indicator, \(Y_i(1)\) and \(Y_i(0)\) the potential outcomes, and \(X_i\in\mathbb R^p\) the covariates, with at least one coordinate being continuously distributed. The observed outcome is $Y_i = T_i Y_i(1) + (1-T_i) Y_i(0).$ Under the superpopulation framework, we assume that
\[
(X_i, Y_i(1), Y_i(0), T_i) \stackrel{\mathrm{i.i.d.}}{\sim} P, \qquad i=1,\dots,n,
\]
where \(P\) is an unknown joint distribution on \((X, Y(1), Y(0), T)\). Let \(\mathcal O_n = \{(X_i,Y_i,T_i)\}_{i=1}^n\) denote the observed data, and let \(\overline{\mathcal O}_n = \{(X_i,Y_i(1),Y_i(0),T_i)\}_{i=1}^n\) denote the corresponding full data. Our estimand of interest is the individual treatment effect (ITE), $\Delta = Y(1) - Y(0)$.

Let \(\mathcal E\) denote the class of all probability distributions on \((X, Y(1), Y(0), T)\) such that at least one coordinate of \(X\) is absolutely continuous with respect to the one-dimensional Lebesgue measure and \(T\) is binary. Define the strong-ignorability class \citep{rosenbaum1983central} by
\[
\mathcal P_{\mathrm{SI}}
= \bigl\{ P \in \mathcal E : (Y(1),Y(0)) \perp T \mid X \bigr\}.
\]

In causal inference problems, one additionally imposes an overlap condition on the propensity score $e_P(X) := P(T=1\mid X)$ \citep{rosenbaum1983central}. We suppress the subscript $P$ (distribution) and write \(e(X)\) when there is no ambiguity. In the main text, we focus on the strong-overlap class
\[
\mathcal P_{\mathrm{SI}}^{\mathrm{str}}
:=
\Bigl\{P\in\mathcal P_{\mathrm{SI}}:
 e_P(X)\in[c_P,1-c_P]\ \text{a.s. for some } c_P>0
\Bigr\}.
\]
That is, the propensity score is bounded away from $0$ and $1$, with the constant \(c_P\) allowed to depend on \(P\). In Appendix~\ref{app:overlap_classes}, we also consider broader overlap-based classes.

A (possibly randomized) prediction set for the ITE is a mapping
\[
\widehat C_{n,\alpha} : \mathbb R^p \to \gB(\sR):=\{\textup{Borel subsets of }\sR\},
\]
constructed from the observed data \(\mathcal O_n\). In the impossibility results, a randomized procedure may additionally use a seed \(U\sim\mathrm{Unif}(0,1)\), independent of all sampled units, with the same seed distribution under every candidate law \(P\). We write \(\widehat C_{n,\alpha}(x;u)\) when displaying this dependence; otherwise, the seed is suppressed, and probabilities and expectations include its randomness. Consider a new unit from the superpopulation with covariates $X$ and the ITE \(\Delta\). We say that \(\widehat C_{n,\alpha}\) achieves finite-sample and asymptotic distribution-free validity at level \(1-\alpha\) over the class \(\mathcal P_{\mathrm{SI}}^{\mathrm{str}}\) if, respectively,
\[
\sup_{P \in \mathcal P_{\mathrm{SI}}^{\mathrm{str}}}
\mathbb P_P\!\left( \Delta \notin \widehat C_{n,\alpha}(X) \right)
\le \alpha,
\qquad
\limsup_{n\to\infty}\sup_{P \in \mathcal P_{\mathrm{SI}}^{\mathrm{str}}}
\mathbb P_P\!\left( \Delta \notin \widehat C_{n,\alpha}(X) \right)
\le \alpha .
\]

In the standard supervised-learning setting, conformal prediction can achieve the above finite-sample distribution-free validity under exchangeability \citep{shafer2008tutorial,lei2018distribution,10.1561/2200000101,angelopoulos2026theoreticalfoundationsconformalprediction}. Our results below show that this guarantee does not extend to the causal setting: even under strong ignorability and overlap assumptions, distribution-free validity for ITE prediction is impossible without additional structural assumptions.

%% file: Sections/Main_results.tex
\subsection{Impossibility Results under Overlap Assumptions}

\label{sec:overlap}

In this subsection, we state our main results, which show that without additional model assumptions, nontrivial inference with a coverage guarantee for the ITE $\Delta$ is impossible. For this reason, we refer to these results as \textbf{impossibility results}. Here and below, $\mathrm{length}$ denotes Lebesgue measure. We begin with the asymptotic result under strong overlap; the finite-sample counterpart is stated in Appendix~\ref{appendix:finite_impossibility}. If the prediction set only achieves asymptotic distribution-free validity at level $1-\alpha$ over $\mathcal P_{\mathrm{SI}}^{\mathrm{str}}$, we have the following asymptotic impossibility result:

\begin{theorem}[Asymptotic impossibility under strong overlap]
\label{thm:asymptotic_impossibility_overlap}
Suppose a sequence of prediction sets \(\{\widehat C_{n,\alpha}\}_{n\ge1}\) satisfies the asymptotic distribution-free validity at level $1-\alpha$ over $\mathcal P_{\mathrm{SI}}^{\mathrm{str}}$:
\[
\limsup_{n\to\infty}
\sup_{P \in \mathcal P_{\mathrm{SI}}^{\mathrm{str}}}
\mathbb P_P\!\left( \Delta \notin \widehat C_{n,\alpha}(X) \right)
\le \alpha .
\]
Then, for any \(P \in \mathcal P_{\mathrm{SI}}\) and any \(y \in \mathbb R\),
\[
\limsup_{n\to\infty}
\mathbb P_P\!\left( y \notin \widehat C_{n,\alpha}(X) \right)
\le \alpha,
\]
and hence
\[
\lim_{n\to\infty}
\mathbb E_P\!\left[ \mathrm{length}\bigl(\widehat C_{n,\alpha}(X)\bigr) \right]
= \infty .
\]
\end{theorem}

The asymptotic theorem shows that any prediction set achieving distribution-free validity for $\Delta$ at level $1-\alpha$ must have asymptotic inclusion probability at least $1-\alpha$ for each pre-specified real number, and hence must have diverging expected length. Appendix~\ref{appendix:finite_impossibility} gives the corresponding finite-sample version. \looseness=-1

\noindent\textit{Proof sketch.} Given the full-data sample $\overline{\mathcal O}_{n+1}$ and the seed $U$, let $\eta_n(\overline{\mathcal O}_{n+1},U)$ be the indicator that $Y_{n+1}(1)-Y_{n+1}(0)\notin\widehat C_{n,\alpha}(X_{n+1};U)$.
\begin{itemize}
    \item \emph{Remove overlap by perturbing treatment.} For each $n$, replace treatment with an independent fair coin flip with probability $\varepsilon$. Letting $\varepsilon\downarrow0$ gives $\sup_{P\in\mathcal P_{\mathrm{SI}}^{\mathrm{str}}}\mathbb E_{P,U}[\eta_n]=\sup_{P\in\mathcal P_{\mathrm{SI}}}\mathbb E_{P,U}[\eta_n]$ (Lemma~\ref{lem:denseness_overlap_subclass}).
    \item \emph{Remove the ignorability restriction using a continuous covariate.} For each fixed $n$, an absolutely continuous coordinate of $X$ permits partitions into arbitrarily many equal-probability cells. Lemma~\ref{lem:nofreelunch} uses these cells to construct a mixture of $(n+1)$-sample laws from $\mathcal P_{\mathrm{SI}}$ that agrees with $Q^{\otimes(n+1)}$ whenever the sampled covariates fall in distinct cells. The agreement is preserved after adjoining the independent seed. This event becomes arbitrarily likely, so Lemma~\ref{lem:asymp_nofreelunch} and the assumed bound give $\limsup_{n\to\infty}\mathbb E_{Q,U}[\eta_n]\le\alpha$ for every fixed $Q\in\mathcal E$.
    \item \emph{Express exclusion of $y$ as miscoverage.} For fixed $P\in\mathcal P_{\mathrm{SI}}$ and $y\in\mathbb R$, complete the missing potential outcomes to obtain $Q_y\in\mathcal E$ with the same observed-data law as $P$ and $\Delta=y$ almost surely. Thus, for every $n$, $\mathbb P_P\{y\notin\widehat C_{n,\alpha}(X)\}=\mathbb P_{Q_y}\{\Delta\notin\widehat C_{n,\alpha}(X)\}=\mathbb E_{Q_y,U}[\eta_n]$. The bound from the second step gives the pointwise inclusion claim.
    \item The expected-length conclusion then follows from Tonelli's theorem and Fatou's lemma.
\end{itemize}

\begin{samepage}
\begin{remark}
In Proposition~\ref{prop:no_covariate_length} of Appendix~\ref{app:overlap_classes}, we prove that infinite expected length holds without any regularity assumption on $X$, even when there are no covariates, provided coverage holds uniformly over the corresponding strong-overlap class. The proof idea is that on the positive-probability event that all $n$ units are control units, we can assign the unobserved potential outcomes so that the ITE equals any prescribed $y$. Uniform coverage and strong overlap then give a positive inclusion-probability bound that holds uniformly in $y$.
\end{remark}
\end{samepage}

\subsection{A Positive Result under Additional Structural Assumptions}

The impossibility results above show that strong ignorability and overlap alone do not suffice for distribution-free predictive inference for the ITE with finite expected length, even in the absence of covariates. A natural way to obtain informative intervals is to impose additional structural assumptions on the data-generating process. We illustrate this in a simple, admittedly restrictive setting, not to claim broad applicability, but to show that informative prediction is possible beyond the minimal causal assumptions. Specifically, we consider an additive model in which the random error components are independent of the covariates and the two potential outcomes are conditionally independent given the covariates. \looseness=-1

For $t\in\{0,1\}$, let
\[
\mu_t(x):=\mathbb E[Y(t)\mid X=x], \qquad
\tau(x):=\mu_1(x)-\mu_0(x), \qquad
\varepsilon_t:=Y(t)-\mu_t(X).
\]
Denote the treated and control indices by
\[
\mathcal T:=\{i:T_i=1\},\qquad \mathcal C:=\{i:T_i=0\}.
\]

\begin{assumption}\label{ass:positive_structure}
Suppose that:
\begin{enumerate}
    \item[(i)] The joint distribution of $(X,Y(1),Y(0),T)$ satisfies strong ignorability and strong overlap;
    \item[(ii)] for $t\in\{0,1\}$, the residuals and the potential outcomes satisfy: $\varepsilon_t\perp X,$ $Y(1)\perp Y(0)\mid X$. 
\end{enumerate}
\end{assumption}

Under Assumption~\ref{ass:positive_structure},
let $\xi:=\varepsilon_1-\varepsilon_0$. Then
\[
\Delta=Y(1)-Y(0)=\tau(X)+\xi,
\]
and $\xi$ is independent of $X$. This suggests a splitting strategy: estimate $\tau(\cdot)$ on one part and recover the distribution of $\xi$ on the other. We formalize this construction in Algorithm~\ref{alg:split_pair}. Under the assumptions above, the following result establishes the validity of Algorithm~\ref{alg:split_pair}.\looseness=-1

\begin{algorithm}[H]
\caption{Split-and-pair prediction interval for the ITE}
\label{alg:split_pair}
\begin{algorithmic}[1]
\Require Data $\{(X_i,T_i,Y_i)\}_{i=1}^n$, level $\alpha\in(0,1)$, split ratio $\rho_n\in(0,1)$.
\State Sample without replacement subsets $\mathcal T_1\subset\mathcal T$ and $\mathcal C_1\subset\mathcal C$ of sizes $\lfloor \rho_n|\mathcal T|\rfloor$ and $\lfloor \rho_n|\mathcal C|\rfloor$, respectively; define $\mathcal I_1:=\mathcal T_1\cup \mathcal C_1$, $\mathcal T_2:=\mathcal T\setminus \mathcal T_1$, and $\mathcal C_2:=\mathcal C\setminus \mathcal C_1$.
\State Using $\mathcal I_1$, estimate $\hat\mu_1(\cdot)$ and $\hat\mu_0(\cdot)$, and set $\hat\tau(x):=\hat\mu_1(x)-\hat\mu_0(x)$.
\State For $i\in\mathcal T_2$ and $j\in\mathcal C_2$, define $\hat\varepsilon_i^{(1)}:=Y_i-\hat\mu_1(X_i)$ and $\hat\varepsilon_j^{(0)}:=Y_j-\hat\mu_0(X_j)$.
\State Let $m_n:=\min\{|\mathcal T_2|,|\mathcal C_2|\}$. Randomly pair $m_n$ treated residuals and $m_n$ control residuals without replacement, yielding pairs $(i_1,j_1),\dots,(i_{m_n},j_{m_n})$, and form $\hat W_k:=\hat\varepsilon_{i_k}^{(1)}-\hat\varepsilon_{j_k}^{(0)}$ for $k=1,\dots,m_n$.
\State Let $\hat q_{\alpha/2}$ and $\hat q_{1-\alpha/2}$ be the empirical $\alpha/2$- and $(1-\alpha/2)$-quantiles of $\{\hat W_k\}_{k=1}^{m_n}$.
\State Output $\hat C_{n,\alpha}(x):=\bigl[\hat\tau(x)+\hat q_{\alpha/2},\ \hat\tau(x)+\hat q_{1-\alpha/2}\bigr]$.
\end{algorithmic}
\end{algorithm}

\begin{samepage}
\begin{theorem}
\label{prop:positive_rate}
Suppose Assumption~\ref{ass:positive_structure} holds. Assume further that:
\begin{enumerate}
    \item[(i)] for each $t\in\{0,1\}$, $\|\hat\mu_t-\mu_t\|_\infty = o_P(1)$;
    \item[(ii)] the distribution of $\xi$ admits a density that is
    continuous, bounded, and strictly positive in neighborhoods of its $\alpha/2$- and $1-\alpha/2$-quantiles;
    \item[(iii)] $n\rho_n\to\infty$ and $n(1-\rho_n)\to\infty$ as $n\to\infty$.
\end{enumerate}

Then
\[
\mathbb P_P\!\left(
\Delta_{n+1}\in \hat C_{n,\alpha}(X_{n+1})
\right)
=1-\alpha+o(1).
\]
\end{theorem}
\end{samepage}

Theorem~\ref{prop:positive_rate} shows that, under the structural assumptions above, the proposed interval is asymptotically valid without being as conservative as some other procedures, as will be seen in Section~\ref{sec:experiments}. Note that the condition in (i) of the theorem is quite mild and is satisfied by standard nonparametric estimators \citep{stone1982optimal,chen2015optimal}; the conditions in (ii) and (iii) are non-restrictive. The conditional independence in Assumption~\ref{ass:positive_structure} is most essential for our positive results. In Appendix~\ref{app:proof_positive_rate}, we further prove a stronger conditional version and show that its coverage error is governed by two terms: the first-stage regression error in estimating the conditional mean functions and the second-stage quantile estimation error in calibrating the paired residual differences. This decomposition highlights the role of the treatment-stratified splitting ratio $\rho_n$, which controls the tradeoff between these two sources of error. We also discuss this tradeoff in more detail under additional smoothness assumptions on $\mu_t$.

%% file: Sections/Other_methods.tex
Our impossibility results provide a new perspective on previous methods for ITE inference. In particular, they show that nontrivial prediction intervals cannot be obtained purely under strong ignorability and overlap alone. Existing procedures should therefore be understood as depending, explicitly or implicitly, on assumptions that restrict the problem sufficiently to permit informative intervals. Since some of the positive results on ITE inference do not seem to have stated all the conditions or subtle requirements explicitly, they may be interpreted as valid distribution-free predictive inference. We focus on two important results below for a complementary and nuanced understanding of the foundational matter.

\subsection{Lei and Cand\`es (2021): conformal inference for counterfactuals and ITEs}

We begin with \cite{lei2021conformal}, which develops a conformal approach to inferring ITEs by treating the missing counterfactual as the primary prediction target. For an in-study unit, one observes only
\[
Y_i=T_iY_i(1)+(1-T_i)Y_i(0),
\]
so only one of \(Y_i(1)\) and \(Y_i(0)\) is available. The method first constructs an interval \(\widehat C_t(x)\) for the missing potential outcome \(Y(t)\) at covariate value \(X=x\) for the group with $T_i=1-t$, and then converts these counterfactual intervals into an interval for \(\Delta_i=Y_i(1)-Y_i(0)\). This yields a concrete conformal framework for inferring counterfactuals, and hence ITEs, for units with exactly one observed potential outcome. The construction uses weighted conformal inference, with weights determined by the propensity score \(e(x)=\mathbb P(T=1\mid X=x)\). When the propensity score is known, as in completely randomized or stratified randomized experiments, the resulting procedure yields finite-sample marginal coverage for counterfactual outcomes, and therefore for the derived ITE intervals, without parametric assumptions on the outcome distribution. In observational studies, where \(e(x)\) must be estimated, the paper establishes a doubly robust approximate validity.

However, the problem changes qualitatively when one moves from in-study units to units outside the study: for a new unit, neither potential outcome is observed. In \cite{lei2021conformal}, the first response is a na\"ive Bonferroni construction based on the counterfactual intervals. Using the first-stage intervals \(\widehat C_1(x)\) and \(\widehat C_0(x)\), it forms the out-of-study ITE interval
\[
\widehat C_{\Delta}^{\mathrm{Bon}}(x)
:=
\widehat C_1(x)-\widehat C_0(x)
=
\{u-v:u\in \widehat C_1(x),\ v\in \widehat C_0(x)\}.
\]
If each counterfactual interval attains coverage level \(1-\alpha/2\), then a Bonferroni argument gives target coverage \(1-\alpha\) for \(\Delta\). But this guarantee already rests on the validity of the first-stage counterfactual intervals, which in \cite{lei2021conformal} are constructed by weighted conformal inference and therefore require either knowledge of, or sufficiently accurate estimation of, the propensity score \(e(x)\). In observational studies, estimating \(e(x)\) can itself be difficult, especially with high-dimensional covariates \citep{damour2021overlap,robins1997coda}. Even when the first-stage coverage is reliable, the direct Bonferroni combination is often highly conservative. To mitigate this, the authors propose nested approaches, including an inexact version that regresses the endpoints of in-study ITE intervals on the covariates and then applies the fitted maps to a new covariate value \(x\). As the authors note, however, this method remains inexact and does not provide a formal coverage guarantee for the ITE of a new unit outside the study. The following numerical example illustrates that there actually can be a systematic discrepancy in coverage regardless of the sample size.

\begin{example}
Suppose
\[
Y(1),Y(0)\stackrel{\mathrm{ind}}{\sim}N(0,1),\qquad T\sim \mathrm{Bernoulli}(1/2),\qquad X\perp (Y(1),Y(0),T).
\]
For an in-study unit, conditional on the observed potential outcome, the ITE has variance \(1\), so the exact \(90\%\) interval uses the \(N(0,1)\) quantile \(z_{0.95}\). Regressing the endpoints of these exact in-study intervals on \(X\) therefore yields, since \(X\) is uninformative, the constant interval \([-z_{0.95}, z_{0.95}]\) for a new unit. But for a new unit outside the study, \(\Delta=Y(1)-Y(0)\sim N(0,2)\), so this interval has coverage only \(2\Phi(z_{0.95}/\sqrt2)-1\approx 0.755\), well below \(0.9\).
\end{example}

In this sense, while the inexact method attempts to extend the first-stage construction, it may provide systematically unreliable results for ITE prediction outside the study. To address this lack of a guarantee, the authors further propose an exact nested method for out-of-study ITE prediction. The method applies a second conformalization step to the induced interval-valued data \((X_i,\widehat C_i)\), where \(\widehat C_i\) is the first-stage ITE interval for an in-study unit, thereby turning the nested construction from a heuristic transfer rule into a procedure with a rigorous coverage guarantee for a new unit.  

More specifically, letting \(C\) denote the interval-valued target for a new unit, the exact nested method intends to construct an interval-valued map \(\widehat{\mathcal C}(x)\) such that
\[
\mathbb P\bigl(C\subset \widehat{\mathcal C}(X)\bigr)\ge 1-\gamma.
\]
Combined with the first-stage guarantee, this would yield the guaranteed coverage for the ITE of a unit outside the study. The validity of the above exact nested method does not follow from overlap and strong ignorability alone: it also depends on sufficiently accurate estimation of the propensity score, and hence on additional conditions on the propensity score function. Moreover, even when these extra assumptions restore validity, both \cite{lei2021conformal} and our numerical evidence suggest that the method remains conservative, sometimes producing intervals longer than the naive Bonferroni construction. More broadly, this highlights the central difficulty of out-of-study ITE prediction: achieving validity without excessive conservativeness generally requires assumptions beyond overlap and strong ignorability alone. \looseness=-1

\subsection{Conformal meta-learners for ITE inference}
\label{sec:conformal_meta_learners}

We briefly review the conformal meta-learner framework of \cite{alaa2023conformal}. At a high level, it combines a meta-learning stage that constructs an estimator $\hat{\tau}(x)$ from the potential-outcome regression functions, or related nuisance components, using a T-learner, S-learner, X-learner, or another plug-in approach \citep{kunzel2019metalearners}, with a conformal calibration stage that turns this estimate into an interval for the ITE at a new covariate value $x$ while remaining robust to misspecification of the nuisance learners.

A key technical ingredient in \cite{alaa2023conformal} is the comparison between an \emph{oracle} nonconformity score $V^*$, defined by the unobserved ITE $\Delta$, and an \emph{estimated} (pseudo) score $V_\varphi$ obtained by replacing $\Delta$ with a pseudo-outcome $\tilde Y$. The pseudo score is the quantity used for conformal calibration: it preserves the functional form of the oracle score while remaining computable from the observed data, and under suitable stochastic ordering conditions relative to the oracle score, \cite{alaa2023conformal} derive a coverage guarantee for the resulting conformal pseudo-intervals $\widehat C_\varphi(X_{n+1})$.

\paragraph{Theorem (Alaa et al., 2023).}
Assuming exchangeability of $(X_i,T_i,Y_i(0),Y_i(1))_{i=1}^{n+1}$, there exists a
distribution-dependent threshold $\alpha^\ast \in (0,1)$ such that the pseudo-interval
$\widehat C_\varphi(X_{n+1})$ satisfies
\[
\mathbb P\!\left( Y_{n+1}(1)-Y_{n+1}(0) \in \widehat C_\varphi(X_{n+1}) \right) \ge 1-\alpha,
\qquad \forall \alpha \in (0,\alpha^\ast),
\]
provided that at least one of the following stochastic ordering conditions holds:
(i) $V_\varphi \succeq_{(1)} V^\ast$,
(ii) $V_\varphi \preceq_{(2)} V^\ast$, and
(iii) $V_\varphi \succeq_{\mathrm{mcx}} V^\ast$.
Under condition (i), we have $\alpha^* =1$.

For completeness, the formal definitions of these three order relations are collected in Appendix~\ref{app:stochastic_orders}. Indeed, were the dominance to hold globally, it would yield nontrivial prediction intervals with distribution-free validity at any pre-specified level $\alpha$, contradicting the hardness results of Section~\ref{sec:main_results}.

\begin{corollary}
\label{cor:failure_fosd}
Suppose that, for all distributions $P \in \mathcal P_{\mathrm{SI}}^{\mathrm{str}}$, the pseudo score $V_\varphi$ satisfies
the first-order stochastic dominance
$V_\varphi \succeq_{(1)} V^\ast$.
Then the conformal pseudo-interval $\widehat C_\varphi$ satisfies
\[
\mathbb{P}_{P}\bigl(y \notin \widehat{C}_\varphi(X)\bigr) \leq \alpha, \quad \forall P \in \mathcal P_{\mathrm{SI}}^{\mathrm{str}},\ \forall y \in \mathbb{R}.
\]
Consequently, any such procedure must produce prediction sets with infinite expected length, i.e.,
\[
\sup_{P \in \mathcal P_{\mathrm{SI}}^{\mathrm{str}}}
\mathbb E_P\!\left[ \mathrm{length}\bigl(\widehat C_\varphi(X)\bigr) \right] = \infty .
\]
\end{corollary}

This observation bears directly on distribution-dependent validity thresholds, denoted by $\alpha^\ast$. For a fixed data-generating distribution $P$, one may define a maximal level $\alpha^\ast(P)$ below which coverage holds, but our results show that $\alpha^\ast(P)$ cannot be bounded uniformly away from zero over the model class; in particular, there exists a sequence of distributions $\{P_n\}$ such that $\alpha^\ast(P_n)\to 0$. \looseness = -1

\begin{figure}[!t]
\centering
\includegraphics[width=\textwidth]{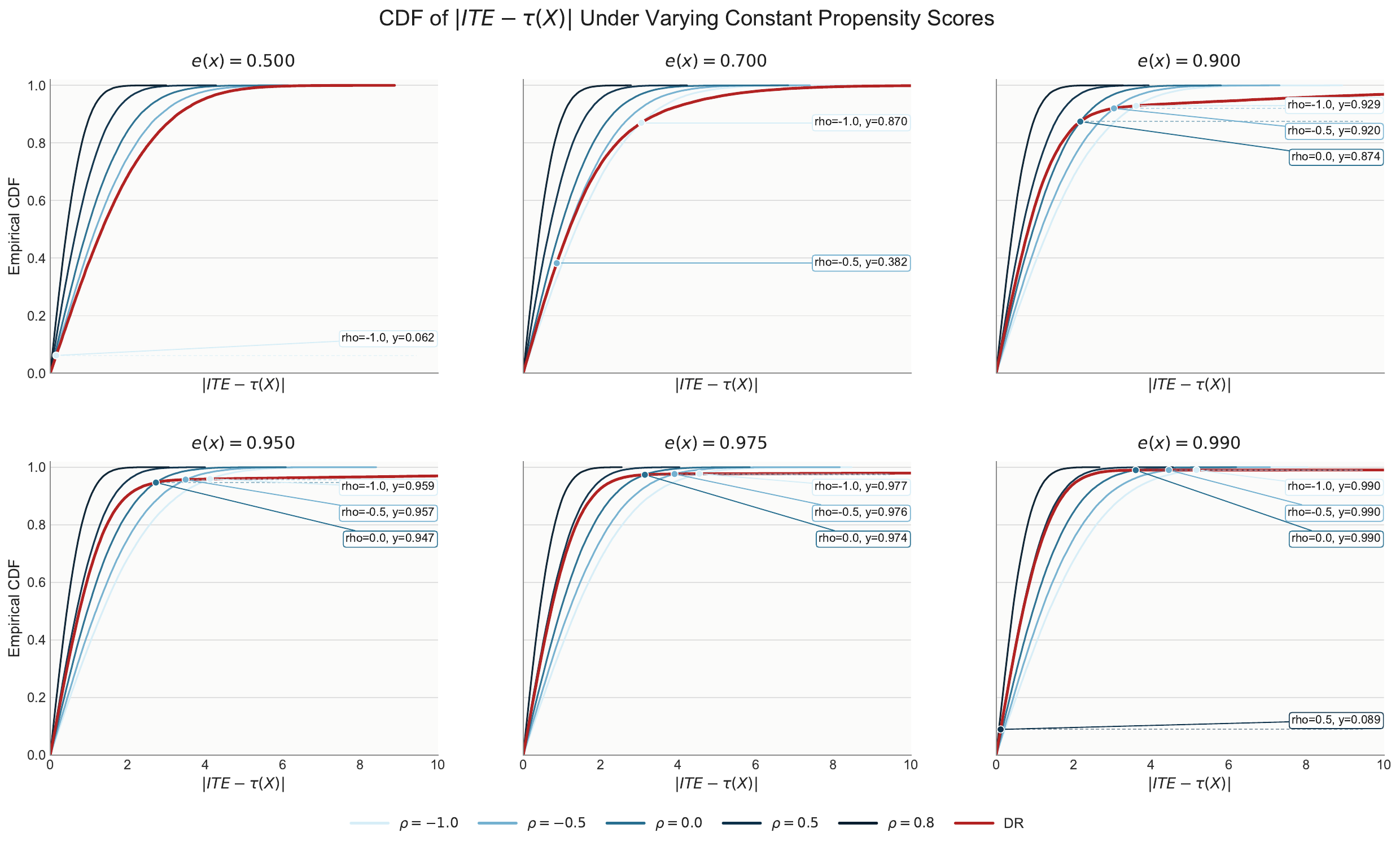}
\caption{CDFs of the DR pseudo score and the true score in a simple setting with known regression functions and known constant propensity scores. The intersections vary substantially across distributions, and the effective validity threshold \(\alpha^\ast\) can be arbitrarily close to \(0\).}
\label{fig:dr_score_cdf_constant_propensity}
\end{figure}

\begin{corollary}
\label{cor:failure_alpha_star}
Suppose that, for all distributions $P \in \mathcal P_{\mathrm{SI}}^{\mathrm{str}}$, the distribution-dependent validity thresholds $\alpha^*(P)$ satisfy
\[
\alpha^*(P) \ge \epsilon >0.
\]
Then for $0 < \alpha \le \epsilon$, the conformal pseudo-interval $\widehat C_\varphi$ satisfies
\[
\mathbb{P}_{P}\bigl(y \notin \widehat{C}_\varphi(X)\bigr) \leq \alpha, \quad \forall P \in \mathcal P_{\mathrm{SI}}^{\mathrm{str}},\ \forall y \in \mathbb{R}.
\]
Consequently, any such procedure must produce prediction sets with infinite expected length, i.e.,
\[
\sup_{P \in \mathcal P_{\mathrm{SI}}^{\mathrm{str}}}
\mathbb E_P\!\left[ \mathrm{length}\bigl(\widehat C_\varphi(X)\bigr) \right] = \infty .
\]
\end{corollary}

From a practical perspective, this sharply limits the usefulness of guarantees of the form “valid for all $\alpha\le\alpha^\ast$.” In applications, users must choose a nominal level $\alpha$ without knowing the distribution-dependent threshold $\alpha^\ast(P)$; our simulations show that even in seemingly benign settings, this effective validity threshold may be arbitrarily close to 0.

To illustrate this point, Figure~\ref{fig:dr_score_cdf_constant_propensity} considers a simple setting with known regression functions and propensity score, and compares the cdf of the DR pseudo score with that of the corresponding true score across different values of the dependence parameter \(\rho\) and different constant propensity scores. Even in this favorable setting, the pseudo score does not uniformly dominate the true score: the two cdfs intersect substantially across distributions and, in some cases, only at very small tail probabilities. Equivalently, the admissible target level, or validity threshold \(\alpha^\ast\), can be arbitrarily small and may approach \(0\). This is fully consistent with Corollary~\ref{cor:failure_fosd} and Corollary~\ref{cor:failure_alpha_star}: the target level is inherently distribution-dependent, so no nontrivial \(\alpha\) can be chosen uniformly over the whole model class.

%% file: Sections/Experiments.tex
\begin{figure}[!b]
\centering
\includegraphics[width=.95\textwidth]{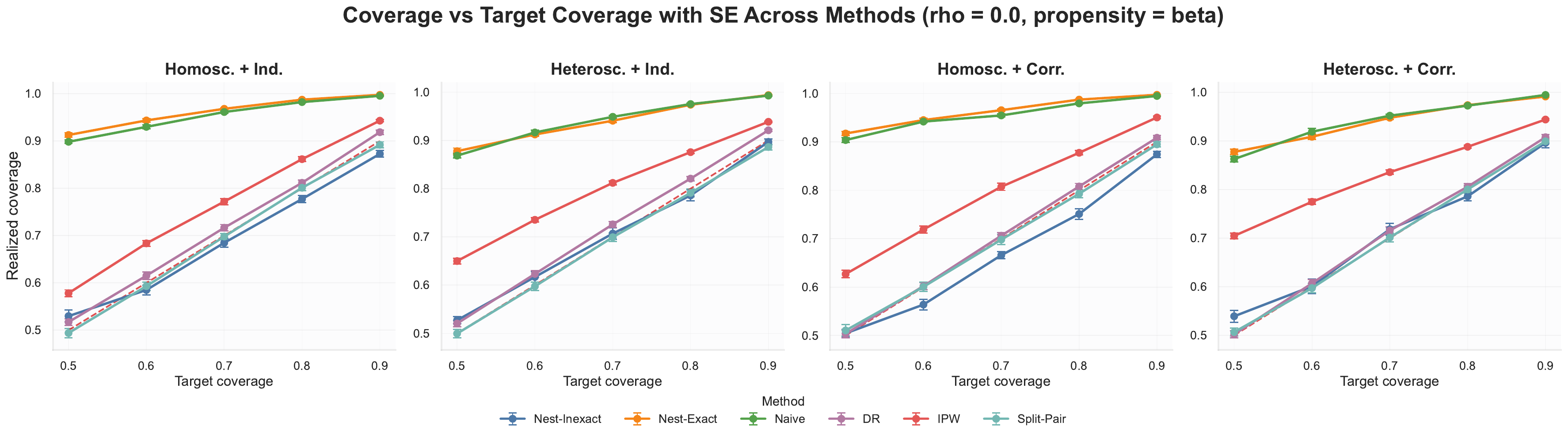}
\vspace{-0.5em}
\includegraphics[width=.95\textwidth]{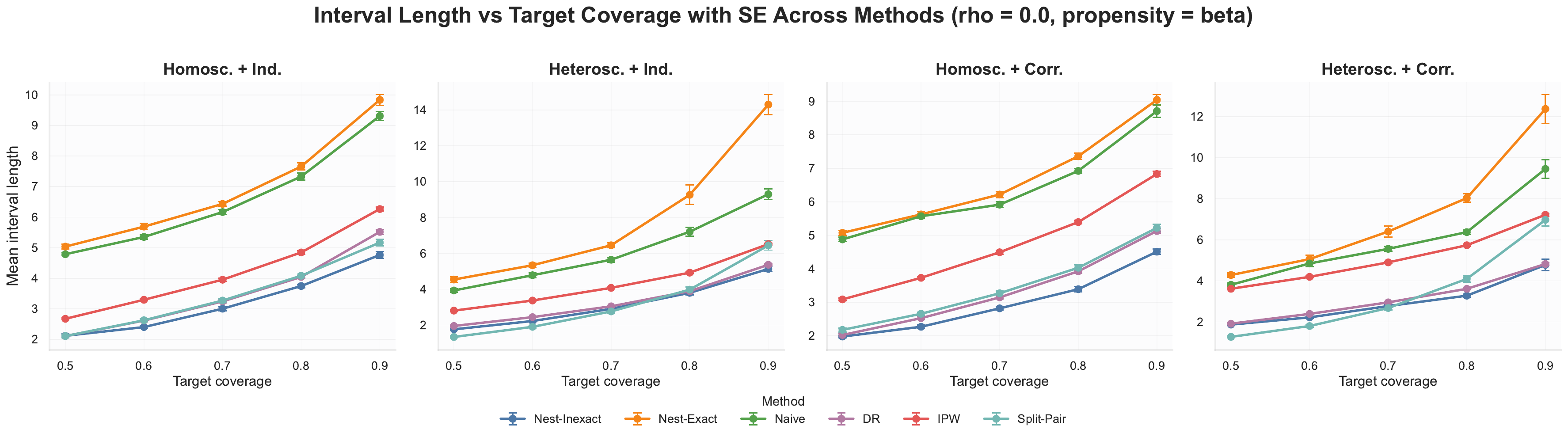}
\caption{Top: realized coverage versus nominal target coverage under the beta propensity score setting with independent errors \(\varepsilon_0\) and \(\varepsilon_1\). The red dashed line is the identity line. Curves below it fail to attain the target coverage. Bottom: average interval length versus nominal target coverage under the same setting.}
\label{fig:target_beta_rho0}
\end{figure}

In this section, we conduct simulation studies to illustrate the practical implications of our theoretical results. We compare several representative methods for predictive inference on individual treatment effects (ITEs), with a focus on empirical coverage and interval length. 

Our simulations are based on the design of \citep{lei2021conformal}, which in turn is adapted from \citep{wager2018estimation}, with a modification tailored to ITE inference. Let \(X=(X_1,X_2)^\top\), where
\[
X_j=\Phi(X_j'), \qquad j=1,2,
\]
\(\Phi\) is the standard normal cdf, and \((X_1',X_2')\) is a bivariate Gaussian vector with mean zero and marginal variance one. We consider both independent and correlated covariates, as in \citep{lei2021conformal}.

The treated potential outcome is generated from
\[
\mathbb{E}[Y(1)\mid X]=f(X_1)f(X_2), \qquad
f(x)=\frac{2}{1+\exp\{-12(x-0.5)\}},
\]
and we set
\[
Y(1)=f(X_1)f(X_2)+\sigma(X)\varepsilon_1, \qquad
Y(0)=\sigma(X)\varepsilon_0.
\]
We consider
\[
\mathrm{Corr}(\varepsilon_0,\varepsilon_1)\in\{-1,-0.5,0\},
\]
to allow $\varepsilon_0$ and $\varepsilon_1$ to be dependent in the first two correlation cases. We focus on negative correlations because they yield a more challenging setting for ITE prediction.

For treatment assignment, we consider two propensity score regimes. In the covariate-dependent setting, following \citep{lei2021conformal}, we set
\[
e(X)=\mathbb{P}(T=1\mid X)=\frac{1}{4}\bigl(1+\beta_{2,4}(X_1)\bigr),
\]
where \(\beta_{2,4}\) is the cdf of the \(\mathrm{Beta}(2,4)\) distribution. In the constant-propensity setting, we take
\[
e(X)=p, \qquad p\in\{0.1,0.3,0.5,0.7,0.9\}.
\]
Together, these regimes let us examine performance under both covariate-dependent treatment assignment and varying degrees of treatment imbalance. In the implementations where the propensity score is treated as unknown, we estimate \(e(X)\) using a boosting-based learner rather than plugging in the oracle propensity score. Additional implementation details are provided in Appendix~\ref{app:cp-ite-simulation-details}.

Unlike \citep{lei2021conformal}, we do not set \(Y(0)=0\). This is because ITE inference is fundamentally different from CATE estimation or counterfactual prediction: one cannot in general assume that the control potential outcome is non-random or independent of \(Y(1)\). Allowing \(Y(0)\) to remain random and correlated with \(Y(1)\) therefore gives a more realistic setting for inferring \(\Delta=Y(1)-Y(0)\).

We begin with the beta propensity score setting \(e(X)=\frac{1}{4}(1+\beta_{2,4}(X_1))\) and independent errors \(\varepsilon_0\) and \(\varepsilon_1\). Figure~\ref{fig:target_beta_rho0} reports realized coverage, together with the corresponding average interval length, for six methods: Naive, Nest-Exact, and Nest-Inexact from \citep{lei2021conformal}; DR and IPW from \citep{alaa2023conformal}; and our Split-Pair method. Several features stand out. Nest-Inexact does not consistently attain the nominal target level, and this failure becomes more pronounced in later settings. Although Nest-Exact enjoys a formal guarantee, it is highly conservative in practice, often even more so than the Naive method. DR and IPW appear somewhat better calibrated in this particular setting, but they remain generally conservative and, as the later experiments show, this behavior is not uniform across settings. By contrast, Split-Pair attains the target coverage level across all four scenarios in this setting. 

We next consider the same beta propensity score setting but allow dependence between the two potential outcomes. Figure~\ref{fig:coverage_beta_rho_neg05} reports the empirical coverage when \(\mathrm{Corr}(\varepsilon_0,\varepsilon_1)=-0.5\). In this setting, Nest-Inexact no longer attains the target level \(0.9\), and this agrees with the example given earlier. We also see that Split-Pair fails to attain the target in this case. This is expected, since \(Y(1)\) and \(Y(0)\) are now correlated, so the setting no longer satisfies the applicability condition of our method. \looseness=-1

\begin{figure}[!t]
\centering
\includegraphics[width=.95\textwidth]{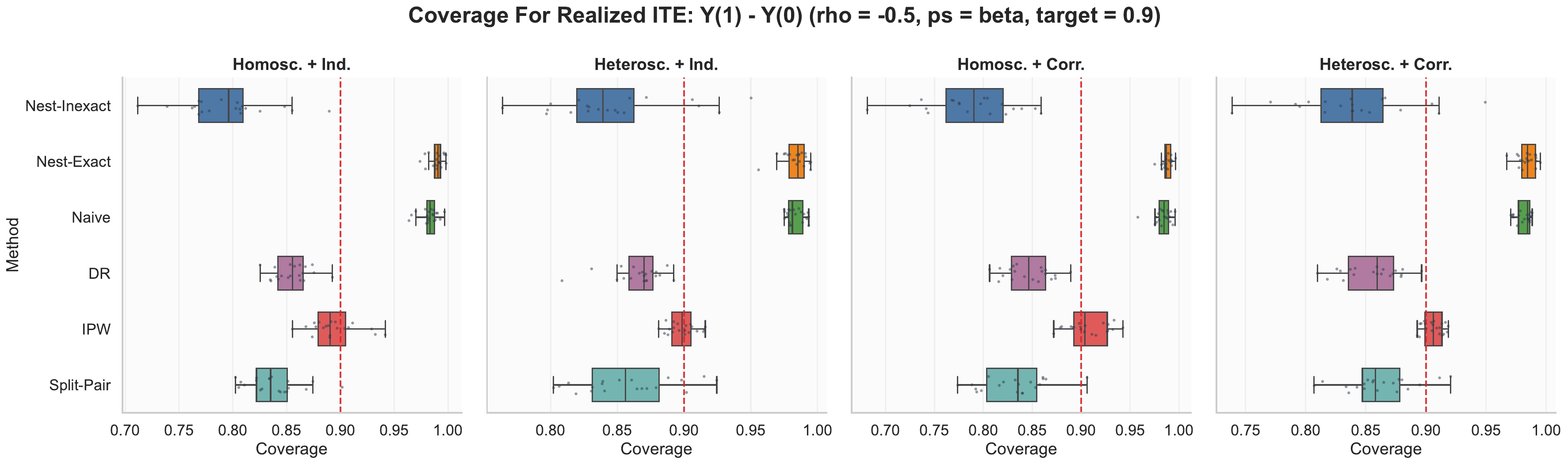}
\caption{Empirical coverage of realized ITE intervals under the beta propensity score setting with \(\mathrm{Corr}(\varepsilon_0,\varepsilon_1)=-0.5\). The red dashed line marks the nominal coverage level \(0.9\).}
\label{fig:coverage_beta_rho_neg05}
\end{figure}

\begin{figure}[!t]
\centering
\includegraphics[width=.95\textwidth]{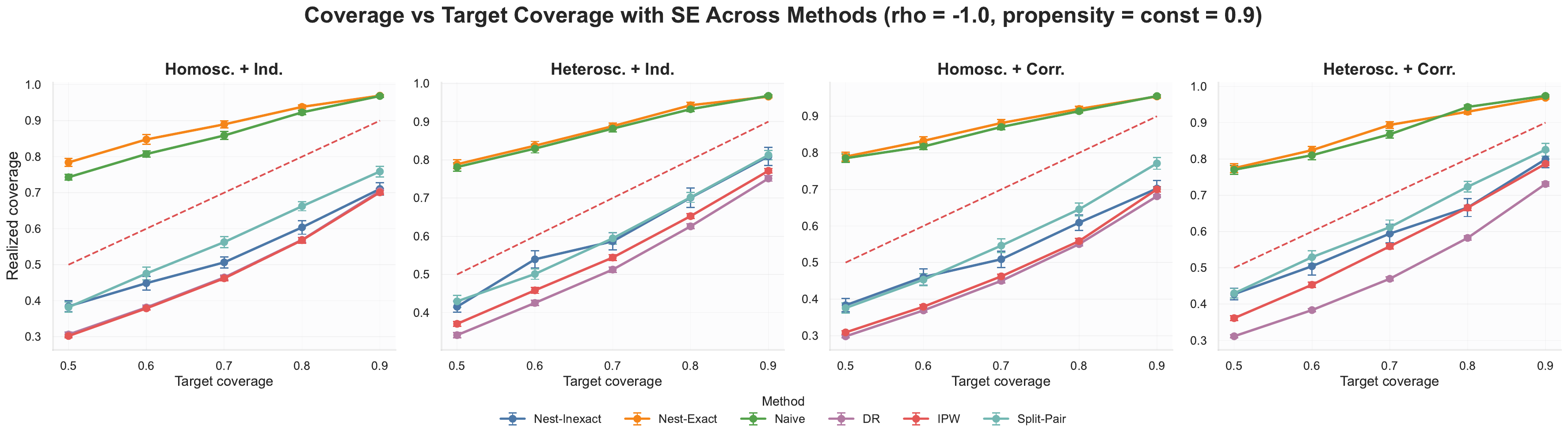}
\caption{Realized coverage versus nominal target coverage when \(\mathrm{Corr}(\varepsilon_0,\varepsilon_1)=-1\) and \(e(X)=0.9\). The red dashed line is the identity line. Curves below it fail to attain the target coverage.}
\label{fig:target_const09_rhoneg1}
\end{figure}
Figure~\ref{fig:target_const09_rhoneg1} considers a more challenging setting with \(\mathrm{Corr}(\varepsilon_0,\varepsilon_1)=-1\) and constant propensity score \(e(X)=0.9\), and plots realized coverage against the nominal target level. In this regime, Nest-Inexact, DR, and IPW all fail to track the target level closely, with realized coverage remaining well below the diagonal over a wide range of targets. This is consistent with Corollary~\ref{cor:failure_fosd} and Corollary~\ref{cor:failure_alpha_star}: for conformal meta-learners, the nominal target level cannot be chosen uniformly over the model class, since validity necessarily depends on the underlying distribution and no single nontrivial \(\alpha\) yields a distribution-free guarantee across all data-generating processes. Thus, even when DR and IPW appear reasonably calibrated in some settings, this behavior cannot persist uniformly.

We also study a substantially harder stress-test design, described in detail in Appendix~\ref{app:checkerboard_stress_test}. In this design, the covariate space is partitioned into a fine checkerboard, and the propensity score alternates sharply between high- and low-probability cells. This makes estimation of \(e(X)\) by default learners much more difficult in finite samples. Figure~\ref{fig:checkerboard_hard_setting} shows that in this regime all methods fall materially below the identity line across the target levels we consider. Most notably, the conservative behavior of Naive and Nest-Exact seen in the earlier experiments no longer persists: once the propensity structure becomes sufficiently difficult, even these methods fail to attain the nominal target coverage. This also highlights the difficulty of distribution-free predictive inference for individual treatment effects.

\begin{figure}[!b]
\centering
\includegraphics[width=.95\textwidth]{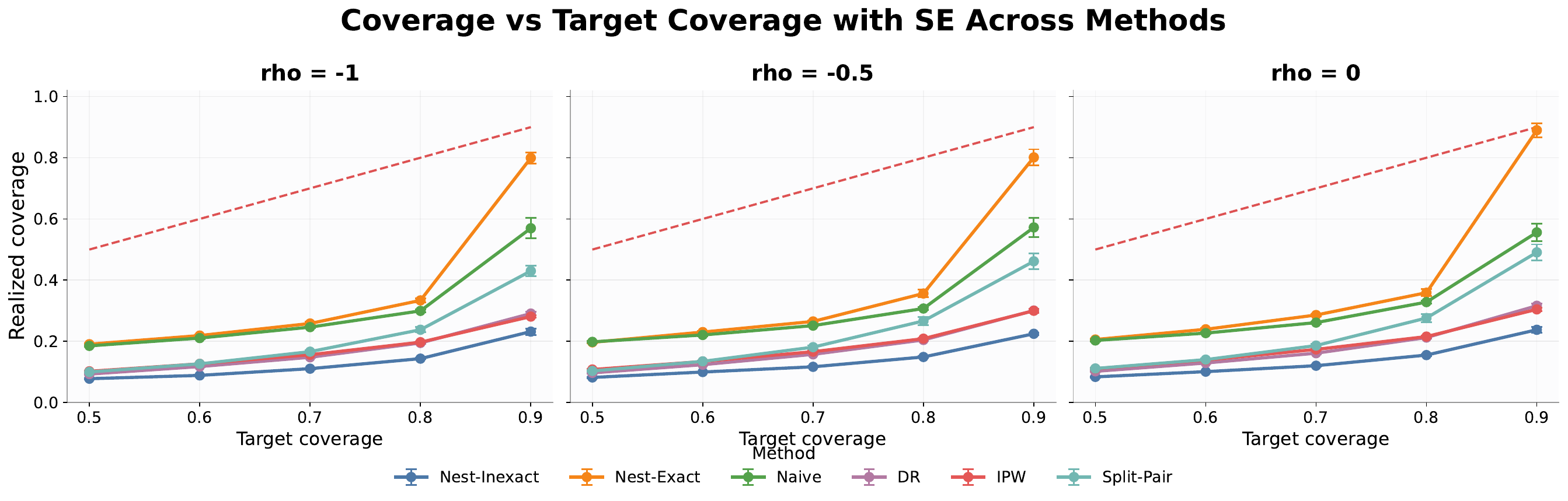}
\caption{Realized coverage versus nominal target coverage in the checkerboard stress-test design. The three panels correspond to \(\rho\in\{-1,-0.5,0\}\). The red dashed line is the identity line.}
\label{fig:checkerboard_hard_setting}
\end{figure}

%% file: Sections/Conclusions.tex
Distribution-free predictive inference has achieved tremendous success in classical regression and related contexts. In this paper, we studied the feasibility of distribution-free predictive inference for ITEs under the standard causal assumptions of strong ignorability and overlap. Our main conclusion is unfortunately negative: any prediction set that attains valid distribution-free coverage for the ITE must be trivial, in the sense of having infinite expected length, and this impossibility holds both in finite samples and asymptotically. This infinite-expected-length conclusion requires no regularity assumption on the covariates and holds even in their absence. Our analysis provides a fresh perspective on existing methods by clarifying that informative conformal procedures for ITE inference must rely on additional structure beyond the standard causal framework. Under additional assumptions, we also obtain a positive result by proposing the Split-Pair method, which yields informative predictive inference for ITEs. Our simulation studies reinforce this point, showing that methods designed to produce shorter intervals may suffer from undercoverage, whereas methods with stronger formal guarantees often become markedly conservative. Taken together, these results show that the difficulty of distribution-free ITE inference is intrinsic to the missing-data nature of the problem. \looseness=-1

%% file: Sections/Proofs.tex

We provide the proofs of the main theorems in Section~\ref{sec:main_results} in the Appendix. The Appendix is organized as follows: 
\begin{itemize}
    \item[A.1.] we prove a technical lemma and state its asymptotic consequence;
    \item[A.2.] we state the impossibility results over $\mathcal{P}_{\textup{SI}}$, which serves as the foundation for the proofs of Theorems~\ref{thm:impossibility_overlap} and~\ref{thm:asymptotic_impossibility_overlap};
    \item[A.3.] we present the proofs of Theorems~\ref{thm:impossibility_overlap} and~\ref{thm:asymptotic_impossibility_overlap} based on the results in A.2;
    \item[A.4.] we present the proofs of the results stated in A.2;
    \item[A.5.] we present the proof of Theorem~\ref{prop:positive_rate} and further discussion on the splitting ratio.
\end{itemize}
\subsection{Technical Lemmas}

Let \(\mathcal D:=\mathbb R^p\times\mathbb R^2\times\{0,1\}\) denote the full-data sample space. A possibly randomized test is a measurable map
\[
\psi_n:\mathcal D^n\times[0,1]\to\{0,1\}.
\]
Here \(U\) is independent of the sample and all auxiliary variables. Writing \(\lambda\) for the law of \(U\), we use \(\mathbb E_{P,U}\) for expectation under \(P^{\otimes n}\otimes\lambda\).

\begin{lemma}
\label{lem:nofreelunch}
Consider any (possibly randomized) test
\[
\psi_n : \mathcal D^n\times[0,1] \to \{0,1\}
\]
based on the full data
$\overline{\mathcal O}_n=\{(X_i,Y_i(1),Y_i(0),T_i)\}_{i=1}^n$.
If
\[
\sup_{P\in\mathcal P_{\mathrm{SI}}} \mathbb E_{P,U}[\psi_n(\overline{\mathcal O}_n,U)] \le \alpha
\quad\text{for some } \alpha\in(0,1),
\]
then it necessarily holds that
\[
\sup_{Q\in\mathcal E} \mathbb E_{Q,U}[\psi_n(\overline{\mathcal O}_n,U)] \le \alpha.
\]
\end{lemma}

\begin{proof}
Fix $Q\in\mathcal E$ and choose a coordinate $X^{(r)}$ with an absolutely continuous marginal law under $Q$. For any positive integer $k$, its continuous distribution function partitions the covariate space into $k$ Borel cells $A_1,\ldots,A_k$, each with $Q_X$-probability $1/k$. Let $b(x)$ be the cell index. Write $q(x)=Q(T=1\mid X=x)$, $q_1(x)=q(x)$, and $q_0(x)=1-q(x)$. Since the sample space is standard Borel, choose regular conditional kernels $K_t(x,\cdot)=Q((Y(1),Y(0))\in\cdot\mid X=x,T=t)$; their values where $q_t(x)=0$ are arbitrary.

For $v=(v_1,\ldots,v_k)\in(0,1)^k$, define $g_v(x)=\mathbf1\{v_{b(x)}<q(x)\}$. Construct $P_v$ by drawing $X\sim Q_X$, setting $T=g_v(X)$, and drawing $(Y(1),Y(0))\sim K_{g_v(X)}(X,\cdot)$ conditionally on $X$. Then $P_v\in\mathcal P_{\mathrm{SI}}$.

Define the probability measure $M_k$ on $\mathcal D^n$ by $M_k(B)=\int_{(0,1)^k}P_v^{\otimes n}(B)\,dv$ for each measurable $B\subseteq\mathcal D^n$. We have $\mathbb E_{M_k\otimes\lambda}[\psi_n]=\int_{(0,1)^k}\mathbb E_{P_v^{\otimes n}\otimes\lambda}[\psi_n]\,dv\leq\alpha$.

Let $D_k=\{b(X_i)\ne b(X_j)\text{ for all }1\leq i<j\leq n\}$. For $x_1,\ldots,x_n$ in distinct cells and $t_1,\ldots,t_n\in\{0,1\}$,
\[
\int_{(0,1)^k}\prod_{i=1}^n\mathbf1\{g_v(x_i)=t_i\}\,dv
=\prod_{i=1}^n q_{t_i}(x_i).
\]
It follows that $M_k(B\cap D_k)=Q^{\otimes n}(B\cap D_k)$ for every measurable $B\subseteq\mathcal D^n$. Consequently, $M_k\otimes\lambda$ and $Q^{\otimes n}\otimes\lambda$ agree on $D_k\times[0,1]$.

By the union bound, $Q^{\otimes n}(D_k^c)\leq {n\choose2}/k$, since $Q^{\otimes n}\{b(X_i)=b(X_j)\}=\sum_{\ell=1}^k Q_X(A_\ell)^2=1/k$ for each $i<j$. Hence, using $0\leq\psi_n\leq1$,
\[
\begin{aligned}
\mathbb E_{Q^{\otimes n}\otimes\lambda}[\psi_n]
&=\mathbb E_{M_k\otimes\lambda}[\psi_n\mathbf1_{D_k}]
  +\mathbb E_{Q^{\otimes n}\otimes\lambda}[\psi_n\mathbf1_{D_k^c}]\\
&\leq\mathbb E_{M_k\otimes\lambda}[\psi_n]+Q^{\otimes n}(D_k^c)
\leq\alpha+\frac{{n\choose2}}{k}.
\end{aligned}
\]
Letting $k\to\infty$ gives $\mathbb E_{Q,U}[\psi_n(\overline{\mathcal O}_n,U)]\leq\alpha$. Taking the supremum over $Q$ proves the claim.
\end{proof}

The following asymptotic statement is a direct consequence of Lemma~\ref{lem:nofreelunch}.

\begin{lemma}
\label{lem:asymp_nofreelunch}
Let $\{\psi_n\}_{n\ge1}$ be a sequence of (possibly randomized) tests with
$\psi_n : \mathcal D^n\times[0,1] \to \{0,1\}$.
Then
\[
\sup_{Q\in\mathcal E}
\limsup_{n\to\infty} \mathbb E_{Q,U}[\psi_n(\overline{\mathcal O}_n,U)]
\;\le\;
\limsup_{n\to\infty}
\sup_{P\in\mathcal P_{\mathrm{SI}}} \mathbb E_{P,U}[\psi_n(\overline{\mathcal O}_n,U)].
\]
\end{lemma}






\subsection{Impossibility Results over overlap sub-classes of \(\mathcal P_{\mathrm{SI}}\)}
\label{app:overlap_classes}

In this appendix, we also consider the following overlap-based sub-classes of \(\mathcal P_{\mathrm{SI}}\):
\begin{align*}
\mathcal P_{\mathrm{SI}}^{\mathrm{ov}}
&:=
\Bigl\{P\in\mathcal P_{\mathrm{SI}}: 0<e_P(X)<1 \ \text{a.s.}\Bigr\},
\\[6pt]
\mathcal P_{\mathrm{SI}}^{(r)}
&:=
\Bigl\{P\in\mathcal P_{\mathrm{SI}}:
\mathbb E_P\!\bigl[e_P(X)^{-r}(1-e_P(X))^{-r}\bigr] < +\infty
\Bigr\},\quad r\in(0,+\infty),
\\[6pt]
\mathcal P_{\mathrm{SI}}^{\mathrm{str}}
&:=
\Bigl\{P\in\mathcal P_{\mathrm{SI}}:
 e_P(X)\in[c_P,1-c_P]\ \text{a.s. for some } c_P>0
\Bigr\}.
\end{align*}
Here \(\mathcal P_{\mathrm{SI}}^{\mathrm{ov}}\) is the weakest overlap condition. The class \(\mathcal P_{\mathrm{SI}}^{(r)}\) requires a finite inverse moment of \(e(X)(1-e(X))\); in particular, the assumptions used in \cite{lei2021conformal} for counterfactual prediction are equivalent to \(P\in\mathcal P_{\mathrm{SI}}^{(1)}\). Finally, \(\mathcal P_{\mathrm{SI}}^{\mathrm{str}}\) is the usual strong-overlap condition, with the constant \(c_P\) allowed to depend on \(P\). These classes are nested: for any \(0<r_1<r_2<+\infty\),
\[
\mathcal P_{\mathrm{SI}}^{\mathrm{str}}\subset \mathcal P_{\mathrm{SI}}^{(r_2)}\subset \mathcal P_{\mathrm{SI}}^{(r_1)}\subset \mathcal P_{\mathrm{SI}}^{\mathrm{ov}}.
\]
Throughout this subsection, we fix any \(\gP\in\{\mathcal P_{\mathrm{SI}}^{\mathrm{ov}},\mathcal P_{\mathrm{SI}}^{\mathrm{str}},\mathcal P_{\mathrm{SI}}^{(r)}\}\), and state the impossibility results over \(\gP\); their proofs are deferred to Section~A.4.

Throughout this subsection, let \(\mathcal O_n^{\rm obs}:=\{(X_i,T_i,Y_i)\}_{i=1}^n\) denote the observed data, where \(Y_i=T_iY_i(1)+(1-T_i)Y_i(0)\). We assume that \((\mathcal O_n^{\rm obs},x,u,y)\mapsto \mathbf 1\{y\in \widehat C_{n,\alpha}(x;u)\}\) is jointly measurable, and write \(\mathrm{length}(A)\) for the Lebesgue measure of a measurable set \(A\subseteq\mathbb R\). As in the main text, the seed is suppressed in prediction-set notation unless needed, and all probabilities and expectations include its independent uniform law.

\begin{theorem}[Finite-sample impossibility]
\label{thm:finite_impossibility}
Suppose a prediction set \(\widehat C_{n,\alpha}\) satisfies the uniform finite-sample coverage guarantee
\[
\sup_{P \in \mathcal P_{\mathrm{SI}}}
\mathbb P_P\!\left( \Delta \notin \widehat C_{n,\alpha}(X) \right)
\le \alpha
\]
for some \(\alpha \in (0,1)\).
Then, for any \(P \in \mathcal P_{\mathrm{SI}}\) and any \(y \in \mathbb R\),
\[
\mathbb P_P\!\left( y \notin \widehat C_{n,\alpha}(X) \right) \le \alpha.
\]
Consequently,
\[
\mathbb E_P\!\left[ \mathrm{length}\bigl(\widehat C_{n,\alpha}(X)\bigr) \right] = \infty .
\]
\end{theorem}

Theorem~\ref{thm:finite_impossibility} shows that any prediction set achieving uniform finite-sample
validity under ignorability assumptions must, with probability at least \(1-\alpha\), contain each pre-specified real
number. As a result, such sets are necessarily uninformative, having infinite expected length.

\begin{theorem}[Asymptotic impossibility]
\label{thm:asymptotic_impossibility}
Suppose a sequence of prediction sets \(\{\widehat C_{n,\alpha}\}_{n\ge1}\) satisfies
\[
\limsup_{n\to\infty}
\sup_{P \in \mathcal P_{\mathrm{SI}}}
\mathbb P_P\!\left( \Delta \notin \widehat C_{n,\alpha}(X) \right)
\le \alpha .
\]
Then, for any \(P \in \mathcal P_{\mathrm{SI}}\) and any \(y \in \mathbb R\),
\[
\limsup_{n\to\infty}
\mathbb P_P\!\left( y \notin \widehat C_{n,\alpha}(X) \right)
\le \alpha,
\]
and hence
\[
\lim_{n\to\infty}
\mathbb E_P\!\left[ \mathrm{length}\bigl(\widehat C_{n,\alpha}(X)\bigr) \right]
= \infty .
\]
\end{theorem}

Theorem~\ref{thm:asymptotic_impossibility} establishes that even asymptotic distribution-free validity
under ignorability assumptions forces prediction sets to be uninformative. The impossibility therefore
reflects the hardness of prediction for ITE under the distribution-free assumption.

\begin{remark}
Our results apply to procedures whose prediction sets depend only on the observed data, the new unit's covariates, and independent auxiliary randomization. In randomized experiments, the known propensity score can be used as an additional input, as in \citet{lei2021conformal} and the IPW and DR methods of \citet{alaa2023conformal}. Our impossibility results do not directly apply to those procedures when they use this additional input.
\end{remark}

Let $\widetilde{\mathcal P}_{\mathrm{SI}}^{\mathrm{str}}$ be the following class, with no restriction on the marginal law of $X$:
\[
\widetilde{\mathcal P}_{\mathrm{SI}}^{\mathrm{str}}
:=
\left\{P\text{ on }(X,Y(1),Y(0),T):
\begin{array}{l}
X\in\mathbb R^p,\quad Y(1),Y(0)\in\mathbb R,\quad T\in\{0,1\},\\
(Y(1),Y(0))\perp T\mid X,\\
e_P(X)\in[c_P,1-c_P]\quad P\text{-a.s. for some }c_P>0
\end{array}
\right\}.
\]

\begin{proposition}[Infinite expected length without covariate regularity]
\label{prop:no_covariate_length}
Fix $n\ge1$ and $\alpha\in(0,1)$. If
\[
\sup_{P\in\widetilde{\mathcal P}_{\mathrm{SI}}^{\mathrm{str}}}
P\{\Delta\notin\widehat C_{n,\alpha}(X)\}\le\alpha,
\]
then $\mathbb E_P[\mathrm{length}(\widehat C_{n,\alpha}(X))]=\infty$ for every $P\in\widetilde{\mathcal P}_{\mathrm{SI}}^{\mathrm{str}}$.
\end{proposition}

\begin{proof}
Fix $P\in\widetilde{\mathcal P}_{\mathrm{SI}}^{\mathrm{str}}$. For each $y\in\mathbb R$, define $Q_y$ by
\[
\mathcal L_{Q_y}(X,Y(0))=\mathcal L_P(X,Y(0)),
\qquad
Q_y\{Y(1)=Y(0)+y,\ T=0\}=1.
\]
For $\varepsilon\in(0,1/2)$, define $Q_{y,\varepsilon}$ by
\[
\mathcal L_{Q_{y,\varepsilon}}(X,Y(1),Y(0))
=\mathcal L_{Q_y}(X,Y(1),Y(0)),
\qquad
Q_{y,\varepsilon}(T=1\mid X,Y(1),Y(0))=\varepsilon.
\]
Then $Q_{y,\varepsilon}\in\widetilde{\mathcal P}_{\mathrm{SI}}^{\mathrm{str}}$ and $\Delta=y$ almost surely, so
\[
1-\alpha
\le Q_{y,\varepsilon}\{y\in\widehat C_{n,\alpha}(X_{n+1})\}
\le (1-\varepsilon)^n Q_y\{y\in\widehat C_{n,\alpha}(X_{n+1})\}
   +1-(1-\varepsilon)^n.
\]
Letting $\varepsilon\downarrow0$ gives $Q_y\{y\in\widehat C_{n,\alpha}(X_{n+1})\}\ge1-\alpha$.

By strong ignorability, for every $y\in\mathbb R$,
\[
\begin{aligned}
P\{y\in\widehat C_{n,\alpha}(X_{n+1})\}
&\ge P\{T_1=\cdots=T_n=0,\ y\in\widehat C_{n,\alpha}(X_{n+1})\}\\
&=\mathbb E_{Q_y,U}\!\left[
\mathbf1\{y\in\widehat C_{n,\alpha}(X_{n+1})\}
\prod_{i=1}^n\{1-e_P(X_i)\}\right]\\
&\ge c_P^n Q_y\{y\in\widehat C_{n,\alpha}(X_{n+1})\}\\
&\ge c_P^n(1-\alpha).
\end{aligned}
\]
The lower bound is uniform in $y$. Tonelli's theorem therefore gives
\[
\mathbb E_P\!\left[\mathrm{length}\bigl(\widehat C_{n,\alpha}(X_{n+1})\bigr)\right]
=\int_{\mathbb R}P\{y\in\widehat C_{n,\alpha}(X_{n+1})\}\,dy
=\infty.
\]
\end{proof}

\begin{remark}
The conclusion of Proposition~\ref{prop:no_covariate_length} also holds when $X$ takes values in any measurable space $\mathcal X$ and its marginal law is restricted to any specified family, provided coverage holds uniformly over the corresponding strong-overlap class and membership in the prediction set is jointly measurable in the observed data, covariates, auxiliary seed, and candidate effect. This includes many standard examples: Gaussian marginals, categorical or text-valued covariates, and a singleton covariate space etc. Indeed, $Q_y$ and $Q_{y,\varepsilon}$ preserve the marginal law of $X$, so the proof applies with $\mathbb R^p$ replaced by $\mathcal X$.
\end{remark}

If $\limsup_{n\to\infty}\sup_{P\in\widetilde{\mathcal P}_{\mathrm{SI}}^{\mathrm{str}}}P\{\Delta\notin\widehat C_{n,\alpha}(X)\}\le\alpha<1$, then the expected length is infinite for every $P\in\widetilde{\mathcal P}_{\mathrm{SI}}^{\mathrm{str}}$ for all sufficiently large $n$. The same eventual conclusion holds for the classes described in the preceding remark.

For each training sample size $n$, the following proofs use the randomized test $\eta_n:\mathcal D^{n+1}\times[0,1]\to\{0,1\}$ defined by
\[
\eta_n(\overline{\mathcal O}_{n+1},u)
=\mathbf 1\!\left\{\,Y_{n+1}(1)-Y_{n+1}(0)\notin \widehat C_{n,\alpha}(X_{n+1};u)\,\right\}.
\]
Thus, $\mathbb E_{P,U}[\eta_n(\overline{\mathcal O}_{n+1},U)]$ is the miscoverage probability, averaging over the full-data sample and the independent procedure seed.

\subsection{Proof of Theorem~\ref{thm:impossibility_overlap} and ~\ref{thm:asymptotic_impossibility_overlap}}
\label{appendix:finite_impossibility}

The proofs of the main theorems proceed by verifying the conditions of Theorem~\ref{thm:finite_impossibility} or Theorem~\ref{thm:asymptotic_impossibility}. To this end, we first establish the following technical lemma.
\begin{lemma}
\label{lem:denseness_overlap_subclass}

Let
\[
\psi_n:\mathcal D^n\times[0,1]\to\{0,1\}
\]
be any measurable (possibly randomized) test. Then
\[
\sup_{P\in\gP}\mathbb E_{P,U}[\psi_n]=\sup_{P\in\mathcal P_{\mathrm{SI}}}\mathbb E_{P,U}[\psi_n].
\]
\end{lemma}

\begin{theorem}[Finite impossibility over overlap sub-classes]
\label{thm:impossibility_overlap}
Suppose a prediction set \(\widehat C_{n,\alpha}\) satisfies finite-sample distribution-free validity at level $1-\alpha$ over $\gP$:
\[
\sup_{P \in \gP}
\mathbb P_P\!\left( \Delta \notin \widehat C_{n,\alpha}(X) \right)
\le \alpha
\]
for some \(\alpha \in (0,1)\).
Then, for any \(P \in \mathcal P_{\mathrm{SI}}\) and any \(y \in \mathbb R\),
\[
\mathbb P_P\!\left( y \notin \widehat C_{n,\alpha}(X) \right) \le \alpha.
\]
Consequently,
\[
\mathbb E_P\!\left[ \mathrm{length}\bigl(\widehat C_{n,\alpha}(X)\bigr) \right] = \infty .
\]
\end{theorem}

Before proving Lemma~\ref{lem:denseness_overlap_subclass}, we show how it yields the stated theorems; the proof of the lemma is given subsequently.

\begin{proof}[Proof of Theorem~\ref{thm:impossibility_overlap}]
It suffices to verify the condition in Theorem~\ref{thm:finite_impossibility}. Applying Lemma~\ref{lem:denseness_overlap_subclass} with sample size \(n+1\) to the test function \(\eta_n\) yields
\[
\begin{aligned}
\sup_{P \in \gP}
\mathbb P_P\!\left( \Delta \notin \widehat C_{n,\alpha}(X) \right)
&=\sup_{P \in \gP}\E_{P,U}[\eta_n(\overline{\mathcal O}_{n+1},U)]\\
&=\sup_{P \in \mathcal P_{\mathrm{SI}}}\E_{P,U}[\eta_n(\overline{\mathcal O}_{n+1},U)]\\
&=
\sup_{P \in \mathcal P_{\mathrm{SI}}}
\mathbb P_P\!\left( \Delta \notin \widehat C_{n,\alpha}(X) \right).
\end{aligned}
\]
Since the left-hand side is at most \(\alpha\) by assumption, the desired condition follows.
\end{proof}

\begin{proof}[Proof of Theorem~\ref{thm:asymptotic_impossibility_overlap}]
It suffices to verify the condition in Theorem~\ref{thm:asymptotic_impossibility}. For each \(n\), applying Lemma~\ref{lem:denseness_overlap_subclass} with sample size \(n+1\) to the test function \(\eta_n\) gives
\[
\sup_{P \in \gP}
\mathbb P_P\!\left( \Delta \notin \widehat C_{n,\alpha}(X) \right)
=
\sup_{P \in \mathcal P_{\mathrm{SI}}}
\mathbb P_P\!\left( \Delta \notin \widehat C_{n,\alpha}(X) \right).
\]
Taking the \(\limsup\) as \(n\to\infty\) and using the assumed asymptotic bound over \(\gP\) yields the desired condition.
\end{proof}

We now turn to the proof of the lemma.

\begin{proof}[Proof of Lemma~\ref{lem:denseness_overlap_subclass}]
The inequality
\[
\sup_{P\in\gP}\mathbb E_{P,U}[\psi_n]\le \sup_{P\in\mathcal P_{\mathrm{SI}}}\mathbb E_{P,U}[\psi_n]
\]
is immediate from \(\gP\subset \mathcal P_{\mathrm{SI}}\). It remains to prove the reverse inequality.

Fix any \(P_0\in\mathcal P_{\mathrm{SI}}\). Let
\[
\mathcal O=(X,Y(1),Y(0),T)\sim P_0,
\]
and let \(B\sim \mathrm{Bernoulli}(1/2)\) be independent of \(\mathcal O\).
For \(\varepsilon\in(0,1)\), define an auxiliary Bernoulli variable
\(A_\varepsilon\sim \mathrm{Bernoulli}(\varepsilon)\), independent of
\((\mathcal O,B)\), and set
\[
T_\varepsilon=(1-A_\varepsilon)T + A_\varepsilon B.
\]
Let \(P_{0,\varepsilon}\) denote the law of
\[
\gO_{\varepsilon}:=(X,Y(1),Y(0),T_\varepsilon).
\]

We first verify that \(P_{0,\varepsilon}\in\mathcal P_{\mathrm{SI}}^{\mathrm{str}}\subset \gP\).
Since \((Y(1),Y(0))\perp T\mid X\) under \(P_0\), and \(A_\varepsilon,B\) are
independent of \((X,Y(1),Y(0),T)\), it follows that
\[
(Y(1),Y(0))\perp T_\varepsilon\mid X.
\]
Hence \(P_{0,\varepsilon}\in\mathcal P_{\mathrm{SI}}\).

Next, writing \(e_0(X)=\mathbb P(T=1\mid X)\) under \(P_0\), we have
\[
\mathbb P(T_\varepsilon=1\mid X)
=
(1-\varepsilon)e_0(X)+\frac{\varepsilon}{2},\qquad
\mathbb P(T_\varepsilon=0\mid X)
=
(1-\varepsilon)(1-e_0(X))+\frac{\varepsilon}{2}.
\]
Therefore
\[
\min\!\bigl(\mathbb P(T_\varepsilon=1\mid X),\mathbb P(T_\varepsilon=0\mid X)\bigr)\ge \frac{\varepsilon}{2}\qquad\text{a.s.},
\]
which proves \(P_{0,\varepsilon}\in\mathcal P_{\mathrm{SI}}^{\mathrm{str}}\subset \gP\).

Since
\[
\|P_{0,\varepsilon}-P_0\|_{\mathrm{TV}}\le \mathbb P(T_\varepsilon\neq T)\le \varepsilon,
\]
we have
\[
P_{0,\varepsilon}^{\otimes n}\otimes\lambda\to P_0^{\otimes n}\otimes\lambda
\qquad\text{in total variation as }\varepsilon\downarrow0.
\]
Because \(0\le \psi_n\le 1\), it follows that
\[
\mathbb E_{P_0,U}[\psi_n]
=
\lim_{\varepsilon\downarrow0}\mathbb E_{P_{0,\varepsilon},U}[\psi_n]
\le
\sup_{\widetilde P\in\gP}\mathbb E_{\widetilde P,U}[\psi_n].
\]
Taking the supremum over \(P_0\in\mathcal P_{\mathrm{SI}}\) yields
\[
\sup_{P\in\mathcal P_{\mathrm{SI}}}\mathbb E_{P,U}[\psi_n]
\le
\sup_{\widetilde P\in\gP}\mathbb E_{\widetilde P,U}[\psi_n].
\]
Combining this with the easy direction completes the proof.
\end{proof}

We next give the proofs of the two auxiliary theorems~\ref{thm:finite_impossibility} and~\ref{thm:asymptotic_impossibility}.

\subsection{Proof of Theorems~\ref{thm:finite_impossibility} and~\ref{thm:asymptotic_impossibility}}

\begin{proof}[Proof of Theorem~\ref{thm:finite_impossibility}]
Under any \(P\in\mathcal P_{\mathrm{SI}}\),
\[
\mathbb E_{P,U}\!\left[\eta_n(\overline{\mathcal O}_{n+1},U)\right]
=\mathbb P_P\!\left(Y_{n+1}(1)-Y_{n+1}(0)\notin \widehat C_{n,\alpha}(X_{n+1})\right)
\le \alpha,
\]
so \(\eta_n\) is a valid level-\(\alpha\) test for the corresponding conditional independence null.
By Lemma~\ref{lem:nofreelunch},
it follows that for any \(Q\in\mathcal E\),
\[
\mathbb E_{Q,U}\!\left[\eta_n(\overline{\mathcal O}_{n+1},U)\right]\le \alpha,
\quad\text{i.e.}\quad
\mathbb P_Q\!\left(Y_{n+1}(1)-Y_{n+1}(0)\notin \widehat C_{n,\alpha}(X_{n+1})\right)\le \alpha.
\]

Fix any \(P\in\mathcal P_{\mathrm{SI}}\) and write
\[
P = P_{Y(1),Y(0)\mid X,T}\otimes P_{X,T}.
\]
For an arbitrary \(y\in\mathbb R\), define a distribution \(Q\in\mathcal E\) on the same space by:
\begin{enumerate}
\item The marginal of \((X,T)\) is preserved: \(Q_{X,T}=P_{X,T}\).
\item The conditional distribution of \((Y(1),Y(0))\) given \((X,T)\) is
\[
Q_{Y(1),Y(0)\mid X,T=t}
=
\begin{cases}
\mathcal L\bigl(Y(0)+y,\; Y(0)\bigr),\ \ Y(0)\sim P_{Y(0)\mid X,T=0}, & t=0,\\[3pt]
\mathcal L\bigl(Y(1),\; Y(1)-y\bigr),\ \ Y(1)\sim P_{Y(1)\mid X,T=1}, & t=1,
\end{cases}
\]
where \(\mathcal L(\cdot)\) denotes the law of a random vector.
\end{enumerate}
Under this construction,
\[
Q_{(X,T,Y)}=P_{(X,T,Y)},
\]
and, moreover,
\[
Q_{Y(1)\mid X,T=1}=P_{Y(1)\mid X,T=1},
\qquad
Q_{Y(0)\mid X,T=0}=P_{Y(0)\mid X,T=0},
\]
and
\[
Q_{Y(1)-Y(0)\mid X,T=1}=Q_{Y(1)-Y(0)\mid X,T=0}=\delta_y,
\]
so \(Q\in\mathcal E\) and \(\Delta=Y(1)-Y(0)\equiv y\) almost surely under \(Q\).

Since the samples are i.i.d. under both \(P\) and \(Q\), and \(Q_{(X,T,Y)}=P_{(X,T,Y)}\) together with \(Q_X=P_X\), we have
\[
\mathcal L_Q\{(X_i,T_i,Y_i)_{i=1}^n,X_{n+1},U\}
=
\mathcal L_P\{(X_i,T_i,Y_i)_{i=1}^n,X_{n+1},U\}.
\]

Therefore,
\begin{align*}
\mathbb P_Q\!\left(Y_{n+1}(1)-Y_{n+1}(0)\notin \widehat C_{n,\alpha}(X_{n+1})\right)
&=
\mathbb P_Q\!\left(y\notin \widehat C_{n,\alpha}(X_{n+1})\right).
\end{align*}
Applying the bound \(\mathbb E_{Q,U}[\eta_n(\overline{\mathcal O}_{n+1},U)]\le \alpha\) yields
\[
\mathbb P_Q\!\left(y\notin \widehat C_{n,\alpha}(X_{n+1})\right)\le \alpha.
\]
By the equality of observed-data laws above, it follows that
\[
\mathbb P_Q\!\left(y\notin \widehat C_{n,\alpha}(X_{n+1})\right)
=
\mathbb P_P\!\left(y\notin \widehat C_{n,\alpha}(X_{n+1})\right).
\]
Hence \(\mathbb P_P\!\left(y\notin \widehat C_{n,\alpha}(X)\right)\le \alpha\).
This proves the first claim.

To prove the infinite expected length claim, define \(f_n(y)=\mathbb P_P\bigl(y\in \widehat C_{n,\alpha}(X)\bigr)\).
The previous inequality gives \(f_n(y)\ge 1-\alpha\) for every \(y\in\mathbb R\). Hence
\begin{align*}
\mathbb E_P\!\left[\operatorname{length}\bigl(\widehat C_{n,\alpha}(X)\bigr)\right]
&=
\int_{\mathbb R}\mathbb P_P\!\left(y\in \widehat C_{n,\alpha}(X)\right)\,dy
=
\int_{\mathbb R} f_n(y)\,dy
\ \ge\ \int_{\mathbb R} (1-\alpha)\,dy
=+\infty.
\end{align*}
The proof is complete.
\end{proof}

\begin{proof}[Proof of Theorem~\ref{thm:asymptotic_impossibility}]
By the assumed asymptotic validity,
\[
\begin{aligned}
&\limsup_{n\to\infty}\ \sup_{P\in\mathcal P_{\mathrm{SI}}}
\mathbb E_{P,U}\!\left[\eta_n(\overline{\mathcal O}_{n+1},U)\right]\\
&\qquad=
\limsup_{n\to\infty}\ \sup_{P\in\mathcal P_{\mathrm{SI}}}
\mathbb P_P\!\left(Y_{n+1}(1)-Y_{n+1}(0)\notin \widehat C_{n,\alpha}(X_{n+1})\right)
\le \alpha.
\end{aligned}
\]
Applying Lemma~\ref{lem:asymp_nofreelunch}, with the test sequence reindexed by the full-data sample size $n+1$, we obtain that for any \(Q\in\mathcal E\),
\[
\begin{gathered}
\limsup_{n\to\infty}\ \mathbb E_{Q,U}\!\left[\eta_n(\overline{\mathcal O}_{n+1},U)\right]\le \alpha,\\
\text{i.e.}\quad
\limsup_{n\to\infty}\ \mathbb P_Q\!\left(Y_{n+1}(1)-Y_{n+1}(0)\notin \widehat C_{n,\alpha}(X_{n+1})\right)\le \alpha.
\end{gathered}
\]

Fix \(P\in\mathcal P_{\mathrm{SI}}\) and \(y\in\mathbb R\), and construct \(Q\in\mathcal E\) exactly as in the proof of
Theorem~\ref{thm:finite_impossibility}, so that \(\Delta\equiv y\) a.s.\ under \(Q\).
Then for each \(n\),
\[
\mathbb P_Q\!\left(Y_{n+1}(1)-Y_{n+1}(0)\notin \widehat C_{n,\alpha}(X_{n+1})\right)
=
\mathbb P_Q\!\left(y\notin \widehat C_{n,\alpha}(X_{n+1})\right).
\]
Moreover, the same construction preserves the full observed-data law, so for each \(n\),
\[
\mathbb P_Q\!\left(y\notin \widehat C_{n,\alpha}(X_{n+1})\right)
=
\mathbb P_P\!\left(y\notin \widehat C_{n,\alpha}(X_{n+1})\right).
\]
Therefore,
\[
\limsup_{n\to\infty}\ \mathbb P_P\!\left(y\notin \widehat C_{n,\alpha}(X)\right)=\limsup_{n\to\infty}\ \mathbb P_Q\!\left(y\notin \widehat C_{n,\alpha}(X)\right)\le \alpha.
\]

For the divergence of expected length, set \(f_n(y)=\mathbb P_P\bigl(y\in \widehat C_{n,\alpha}(X)\bigr)\).
Then \(\liminf_{n\to\infty} f_n(y)\ge 1-\alpha\) for every \(y\in\mathbb R\).
By Tonelli's theorem, including expectation over the training sample and any auxiliary randomization, $\mathbb E_P[\operatorname{length}(\widehat C_{n,\alpha}(X))]=\int_{\mathbb R}f_n(y)\,dy$. Fatou's lemma therefore gives
\[
+\infty
=\int_{\mathbb R}\liminf_{n\to\infty}f_n(y)\,dy
\leq\liminf_{n\to\infty}\int_{\mathbb R}f_n(y)\,dy
=\liminf_{n\to\infty}\mathbb E_P\!\left[\operatorname{length}\bigl(\widehat C_{n,\alpha}(X)\bigr)\right].
\]
This proves the stated divergence.
\end{proof}

\subsection{Proof of Theorem~\ref{prop:positive_rate}}
\label{app:proof_positive_rate}

We prove the following stronger conditional statement, which is a slightly stronger version of
Theorem~\ref{prop:positive_rate}.

\begin{proposition}
\label{prop:positive_rate_conditional}
Suppose Assumption~\ref{ass:positive_structure} holds. Let the treatment-stratified split ratio
satisfy $\rho_n\in(0,1)$, $n\rho_n\to\infty$, and $n(1-\rho_n)\to\infty$. Assume moreover
that, for each $t\in\{0,1\}$, $\|\hat\mu_t-\mu_t\|_\infty = O_P(r_{n,t})$ for some
deterministic sequence $r_{n,t}\to 0$, and that the distribution of
$\xi:=\varepsilon_1-\varepsilon_0$ admits a density $f_\xi$ which is continuous and strictly
positive in neighborhoods of its $\alpha/2$- and $(1-\alpha/2)$-quantiles
$q_\ell:=F_\xi^{-1}(\alpha/2)$ and $q_u:=F_\xi^{-1}(1-\alpha/2)$. Let
$m_n:=\min\{|\mathcal T_2|,\,|\mathcal C_2|\}$, and let
\[
\mathcal D_n
:=
\sigma\!\Big(
\{(X_i,T_i,Y_i)\}_{i=1}^n,\ \mathcal T_1,\mathcal C_1,\ (i_1,j_1),\dots,(i_{m_n},j_{m_n})
\Big)
\]
denote the sigma-field generated by the observed sample, the treatment-stratified split, and the
random pairing in the second stage of Algorithm~\ref{alg:split_pair}. Then
\[
\left|
\mathbb P_P\!\left(
\Delta_{n+1}\in \hat C_{n,\alpha}(X_{n+1}) \,\middle|\, \mathcal D_n
\right)
-(1-\alpha)
\right|
=
O_P\!\left(r_{n,0}+r_{n,1}+m_n^{-1/2}\right).
\]
\end{proposition}

\begin{proof}
Set $\tau(x):=\mu_1(x)-\mu_0(x)$, $\hat\tau(x):=\hat\mu_1(x)-\hat\mu_0(x)$, and
$e_n(x):=\hat\tau(x)-\tau(x)$. Also define
$a_n:=\|\hat\mu_1-\mu_1\|_\infty+\|\hat\mu_0-\mu_0\|_\infty$. Then
$\|e_n\|_\infty \le a_n$.

\paragraph{Step 1: Paired residual differences.}
Let $\mathcal T:=\{i\in\{1,\dots,n\}:T_i=1\}$ and
$\mathcal C:=\{i\in\{1,\dots,n\}:T_i=0\}$, and denote
$N_1:=|\mathcal T|$ and $N_0:=|\mathcal C|$. After the treatment-stratified split, define
$\mathcal T_2:=\mathcal T\setminus \mathcal T_1$ and
$\mathcal C_2:=\mathcal C\setminus \mathcal C_1$. For $i\in\mathcal T_2$ and
$j\in\mathcal C_2$, define $\varepsilon_i^{(1)}:=Y_i-\mu_1(X_i)$ and
$\varepsilon_j^{(0)}:=Y_j-\mu_0(X_j)$.

Because $\varepsilon_1\perp T$, the treated residuals
$\{\varepsilon_i^{(1)}:i\in\mathcal T_2\}$ are i.i.d. with the same law as $\varepsilon_1$.
Likewise, because $\varepsilon_0\perp T$, the control residuals
$\{\varepsilon_j^{(0)}:j\in\mathcal C_2\}$ are i.i.d. with the same law as $\varepsilon_0$.
Since different units are independent and $\varepsilon_1\perp \varepsilon_0$, it follows that
conditional on $m_n=m$, the random variables
$W_k:=\varepsilon_{i_k}^{(1)}-\varepsilon_{j_k}^{(0)}$, $k=1,\ldots,m$, are i.i.d. with common
distribution $F_\xi$.

Let $F_m^\circ(x):=\frac{1}{m}\sum_{k=1}^{m}\mathbf 1\{W_k\le x\}$ denote their empirical cdf.
On the event \(\{m_n=m\}\), we write \(F_{m_n}^\circ:=F_m^\circ\).

\paragraph{Step 2: The effective second-stage sample size diverges and has order $(1-\rho_n)n$.}
Because $P\in\mathcal P_{\mathrm{SI}}^{\mathrm{str}}$, there exists $c\in(0,1/2)$ such that
$e(X)=\mathbb P(T=1\mid X)\in[c,1-c]$ almost surely. Hence
$\pi_1:=\mathbb P(T=1)=\mathbb E[e(X)]\in[c,1-c]$ and
$\pi_0:=\mathbb P(T=0)=1-\pi_1\in[c,1-c]$. By the law of large numbers,
$N_1/n\to \pi_1$ and $N_0/n\to \pi_0$ in probability.

Now $|\mathcal T_2|=N_1-\lfloor \rho_n N_1\rfloor$ and
$|\mathcal C_2|=N_0-\lfloor \rho_n N_0\rfloor$. Since $n(1-\rho_n)\to\infty$, we obtain
$|\mathcal T_2|/((1-\rho_n)n)\to \pi_1$ and
$|\mathcal C_2|/((1-\rho_n)n)\to \pi_0$ in probability. Therefore
$m_n/((1-\rho_n)n)\to \min\{\pi_1,\pi_0\}$ in probability. In particular,
$m_n\to\infty$ and $m_n^{-1/2}=O_P\!\left(((1-\rho_n)n)^{-1/2}\right)$.

\paragraph{Step 3: Quantile estimation error.}
Let $\hat W_k:=\bigl(Y_{i_k}-\hat\mu_1(X_{i_k})\bigr)-\bigl(Y_{j_k}-\hat\mu_0(X_{j_k})\bigr)$,
$k=1,\dots,m_n$, and let
$\hat F_{m_n}(x):=\frac{1}{m_n}\sum_{k=1}^{m_n}\mathbf 1\{\hat W_k\le x\}$ denote the
empirical cdf of the plug-in paired differences. Let
$\hat q_\ell:=\hat F_{m_n}^{-1}(\alpha/2)$ and
$\hat q_u:=\hat F_{m_n}^{-1}(1-\alpha/2)$.

For every $k$, $|\hat W_k-W_k|\le \|\hat\mu_1-\mu_1\|_\infty+\|\hat\mu_0-\mu_0\|_\infty=a_n$.
Hence, for every $x\in\mathbb R$,
\[
F_{m_n}^\circ(x-a_n)\le \hat F_{m_n}(x)\le F_{m_n}^\circ(x+a_n).
\]

Write $\delta_n:=\sup_{x\in\mathbb R}|F_{m_n}^\circ(x)-F_\xi(x)|$.
For any integer $m\ge 1$ and any $u>0$, conditional on $m_n=m$, the
Dvoretzky--Kiefer--Wolfowitz inequality gives
\[
\mathbb P\!\left(\sup_{x\in\mathbb R}|F_m^\circ(x)-F_\xi(x)|>u\,\middle|\, m_n=m\right)
\le 2e^{-2mu^2}.
\]
Therefore, for any $M>0$,
$\mathbb P\!\left(\sqrt{m_n}\,\delta_n>M\right)\le 2e^{-2M^2}+\mathbb P(m_n=0)$. Since
$m_n\to\infty$ in probability by Step~2, it follows that $\delta_n=O_P(m_n^{-1/2})$.

Because $f_\xi$ is continuous and strictly positive in neighborhoods of $q_\ell$ and $q_u$, there
exist constants $\underline f>0$, $\bar f<\infty$, and $\epsilon_0>0$ such that
$\underline f\le f_\xi(x)\le \bar f$ for all
$x\in[q_\ell-\epsilon_0,q_\ell+\epsilon_0]\cup[q_u-\epsilon_0,q_u+\epsilon_0]$.
Consequently, for every $x$ in these neighborhoods,
$\left|\hat F_{m_n}(x)-F_\xi(x)\right|\le \delta_n+\bar f a_n$.
Since $F_\xi$ has derivative bounded below by $\underline f$ near $q_\ell$ and $q_u$, the
standard local quantile perturbation argument yields
\[
|\hat q_\ell-q_\ell|+|\hat q_u-q_u|
=
O_P(a_n+\delta_n)
=
O_P\!\left(r_{n,0}+r_{n,1}+m_n^{-1/2}\right).
\]

\paragraph{Step 4: Conditional coverage error.}
By construction, $\hat C_{n,\alpha}(x)=[\hat\tau(x)+\hat q_\ell,\ \hat\tau(x)+\hat q_u]$.
For the new point, $\Delta_{n+1}=\tau(X_{n+1})+\xi_{n+1}$,
where $\xi_{n+1}\sim F_\xi$, $\xi_{n+1}\perp X_{n+1}$, and the test point is independent of
$\mathcal D_n$. Therefore,
\[
\mathbb P_P\!\left(
\Delta_{n+1}\in \hat C_{n,\alpha}(X_{n+1}) \,\middle|\, \mathcal D_n
\right)
=
\mathbb E\!\left[
F_\xi\bigl(\hat q_u+e_n(X_{n+1})\bigr)
-
F_\xi\bigl(\hat q_\ell+e_n(X_{n+1})\bigr)
\,\middle|\, \mathcal D_n
\right].
\]

On the event
$\{|\hat q_\ell-q_\ell|\le \epsilon_0/4,\ |\hat q_u-q_u|\le \epsilon_0/4,\ \|e_n\|_\infty\le \epsilon_0/4\}$,
all arguments of $F_\xi$ above lie inside the neighborhoods on which $f_\xi\le \bar f$.
Hence, uniformly in $x$,
\[
\left|
F_\xi\bigl(\hat q_u+e_n(x)\bigr)-F_\xi(q_u)
\right|
\le
\bar f\bigl(|\hat q_u-q_u|+\|e_n\|_\infty\bigr),
\]
and similarly
\[
\left|
F_\xi\bigl(\hat q_\ell+e_n(x)\bigr)-F_\xi(q_\ell)
\right|
\le
\bar f\bigl(|\hat q_\ell-q_\ell|+\|e_n\|_\infty\bigr).
\]
Since $F_\xi(q_u)-F_\xi(q_\ell)=\left(1-\frac{\alpha}{2}\right)-\frac{\alpha}{2}=1-\alpha$,
it follows that
\begin{align*}
&\left|
\mathbb P_P\!\left(
\Delta_{n+1}\in \hat C_{n,\alpha}(X_{n+1}) \,\middle|\, \mathcal D_n
\right)
-(1-\alpha)
\right| \\
&\qquad\le
\bar f\Bigl(
|\hat q_u-q_u|
+
|\hat q_\ell-q_\ell|
+
2\|e_n\|_\infty
\Bigr).
\end{align*}
Using $\|e_n\|_\infty\le a_n=O_P(r_{n,0}+r_{n,1})$ and the quantile bound established above,
we obtain
\[
\left|
\mathbb P_P\!\left(
\Delta_{n+1}\in \hat C_{n,\alpha}(X_{n+1}) \,\middle|\, \mathcal D_n
\right)
-(1-\alpha)
\right|
=
O_P\!\left(r_{n,0}+r_{n,1}+m_n^{-1/2}\right).
\]
This proves Proposition~\ref{prop:positive_rate_conditional}.
\end{proof}

As an immediate consequence, Proposition~\ref{prop:positive_rate_conditional} implies
Theorem~\ref{prop:positive_rate}. Let
$R_n:=\mathbb P_P\!\left(\Delta_{n+1}\in \hat C_{n,\alpha}(X_{n+1}) \,\middle|\, \mathcal D_n\right)-(1-\alpha)$.
Then $R_n=o_P(1)$ by Proposition~\ref{prop:positive_rate_conditional}, and $|R_n|\le 1$.
Hence bounded convergence gives $\mathbb E[R_n]\to 0$, so
$\mathbb P_P\!\left(\Delta_{n+1}\in \hat C_{n,\alpha}(X_{n+1})\right)=1-\alpha+o(1)$.

The rate in Proposition~\ref{prop:positive_rate_conditional} is governed by two terms:
the first-stage regression error $r_{n,0}+r_{n,1}$ and the second-stage quantile calibration
error $m_n^{-1/2}$. Under strong overlap, $m_n$ is of order $n(1-\rho_n)$ in probability, so
increasing $\rho_n$ typically improves first-stage regression estimation but reduces the
second-stage calibration sample size. For example, under a nonparametric uniform regression
rate of the form $r_{n,t}\asymp (n\rho_n)^{-\beta/(2\beta+d)}$,
the choice of $\rho_n$ reflects the balance between $(n\rho_n)^{-\beta/(2\beta+d)}$ and
$\{n(1-\rho_n)\}^{-1/2}$. This illustrates that the optimal split depends
on the convergence rates of regression estimators and quantile estimators.

%% file: Sections/Stochastic_orders_appendix.tex
\label{app:stochastic_orders}

This appendix records the stochastic-order definitions used in our discussion of conformal meta-learners. More specifically, it provides the formal meanings of first-order stochastic dominance, second-order stochastic dominance, and monotone convex order, following the definitions used in \cite{alaa2023conformal}.

Let $U$ and $V$ be real-valued random variables with cumulative distribution functions $F_U$ and $F_V$.

\begin{definition}[First-order stochastic dominance]
We say that $U$ first-order stochastically dominates $V$, written
\[
U \succeq_{(1)} V,
\]
if
\[
F_U(x) \le F_V(x), \qquad \forall x\in\mathbb R,
\]
with strict inequality for some $x$.
\end{definition}

\begin{definition}[Second-order stochastic dominance]
We say that $U$ second-order stochastically dominates $V$, written
\[
U \succeq_{(2)} V,
\]
if
\[
\int_{-\infty}^x \bigl(F_V(t)-F_U(t)\bigr)\,dt \ge 0,
\qquad \forall x\in\mathbb R,
\]
with strict inequality for some $x$.
\end{definition}

\begin{definition}[Monotone convex order]
We say that $U$ dominates $V$ in monotone convex order, written
\[
U \succeq_{\mathrm{mcx}} V,
\]
if
\[
\mathbb E[u(U)] \ge \mathbb E[u(V)]
\]
for every non-decreasing convex function $u:\mathbb R\to\mathbb R$ for which the expectations are well defined. 
\end{definition}

In the terminology of \cite{alaa2023conformal}, first-order stochastic dominance compares location through the ordering of the cdfs, second-order stochastic dominance compares spread through integrated cdfs, and monotone convex order compares expectations against all non-decreasing convex test functions. These are precisely the three order relations invoked in the stochastic-ordering theorem quoted in Section~\ref{sec:conformal_meta_learners}.

%% file: Sections/Detailed_Exp.tex
\subsection{Experimental Setup}

\label{app:cp-ite-simulation-details}

This appendix gives the full data-generating process and implementation details
for the simulation study.  The design follows
\citet{lei2021conformal}, which is adapted from \citet{wager2018estimation},
with one important modification: both potential outcomes are random, so that the
individual treatment effect is itself a random quantity.

\paragraph{Covariates.}
Let \(X=(X_1,X_2)^\top\).  We generate
\[
    X_j = \Phi(X'_j), \qquad j=1,2,
\]
where \(\Phi\) is the standard normal cdf.  We consider two covariance designs
for \((X'_1,X'_2)\).  In the independent-covariate design,
\((X'_1,X'_2)\sim N(0,I_2)\).  In the correlated-covariate design, we generate
\[
    X'_j = \sqrt{1-\rho_X}\, Z_j + \sqrt{\rho_X}\, Z_0,
    \qquad j=1,2,
\]
where \(Z_0,Z_1,Z_2\) are independent standard normal variables and
\(\rho_X=0.9\).  This transformation produces covariates supported on
\([0,1]^2\), while allowing us to compare independent and strongly correlated
covariate settings.

\paragraph{Potential outcomes.}
Define
\[
    g(x) = \frac{2}{1+\exp\{-12(x-0.5)\}}.
\]
The conditional mean of the treated potential outcome is
\[
    \tau(X) = \mathbb{E}[Y(1)\mid X] = g(X_1)g(X_2),
\]
and the control mean is zero.  The potential outcomes are generated as
\[
    Y(1)=\tau(X)+\sigma(X)\varepsilon_1,
    \qquad
    Y(0)=\sigma(X)\varepsilon_0.
\]
We consider two noise-scale settings:
\[
    \sigma(X)=1
    \quad\text{and}\quad
    \sigma(X)=-\log(X_1).
\]
The first setting is homoscedastic, while the second is heteroscedastic.  In the
implementation, a small numerical offset is used inside the logarithm to avoid
evaluating \(\log(0)\).

The pair \((\varepsilon_0,\varepsilon_1)\) is standard normal marginally, with
\[
    \mathrm{Corr}(\varepsilon_0,\varepsilon_1)
    \in \{-1,-0.5,0\}.
\]
When the correlation is \(-1\), the second error is generated as the
corresponding deterministic multiple of the first.  Otherwise, we generate
\[
    \varepsilon_1
    =
    \rho\,\varepsilon_0 + \sqrt{1-\rho^2}\,\eta,
    \qquad \eta\sim N(0,1),
\]
independently of \(\varepsilon_0\).  We focus on non-positive correlations
because negative dependence between \(Y(0)\) and \(Y(1)\) increases the
variability of the realized ITE \(Y(1)-Y(0)\), making ITE prediction intervals
more challenging.

Combining the two covariate designs and the two noise-scale designs gives four
scenarios:
\[
\begin{array}{ll}
    \text{Scenario 1:} & \text{homoscedastic noise, independent covariates},\\
    \text{Scenario 2:} & \text{heteroscedastic noise, independent covariates},\\
    \text{Scenario 3:} & \text{homoscedastic noise, correlated covariates},\\
    \text{Scenario 4:} & \text{heteroscedastic noise, correlated covariates}.
\end{array}
\]

\paragraph{Treatment assignment.}
Given \(X\), treatment is assigned as
\[
    T\mid X \sim \mathrm{Bernoulli}(e(X)),
\]
and the observed outcome is
\[
    Y = T Y(1) + (1-T)Y(0).
\]
We use two classes of propensity scores.  In the covariate-dependent setting,
following \citet{lei2021conformal}, we set
\[
    e(X)=\frac{1}{4}\{1+\beta_{2,4}(X_1)\},
\]
where \(\beta_{2,4}\) denotes the cdf of a \(\mathrm{Beta}(2,4)\) distribution.
In the constant-propensity setting, we set
\[
    e(X)=p,
    \qquad
    p\in\{0.1,0.3,0.5,0.7,0.9\}.
\]
The constant-propensity settings are included to study the effect of treatment
imbalance.

\paragraph{Difference from the CATE/counterfactual setting.}
In the original simulation design of \citet{lei2021conformal}, the focus is on
counterfactual prediction and CATE-related quantities.  For ITE inference, we do
not set \(Y(0)\) to be deterministic.  Instead, \(Y(0)\) remains random and may
be correlated with \(Y(1)\).  This distinction is important because the target
of inference is the realized individual treatment effect
\[
    \tau_{\mathrm{ITE}} = Y(1)-Y(0),
\]
not the conditional mean contrast \(\mathbb{E}[Y(1)-Y(0)\mid X]\).  Allowing both
potential outcomes to be random therefore creates a setting closer to the
inferential target of individual-level treatment-effect prediction intervals.

\paragraph{Sample sizes and repetitions.}
For each experimental configuration, we generate a training sample of size
\[
    n_{\mathrm{train}}=1000
\]
and an independent test sample of size
\[
    n_{\mathrm{test}}=10000.
\]
Each configuration is repeated \(20\) times using independent random seeds.  We
evaluate target coverages
\[
    0.5,\;0.6,\;0.7,\;0.8,\;0.9.
\]

\paragraph{Methods.}
We compare the following prediction-interval methods.
\begin{itemize}
    \item \textbf{Nest-Inexact}: the inexact nested conformal ITE method
    implemented by \texttt{cfcausal::conformalIte} with
    \(\texttt{algo = ``nest''}\), \(\texttt{exact = FALSE}\), and CQR.
    \item \textbf{Nest-Exact}: the exact nested conformal ITE method implemented
    by \texttt{cfcausal::conformalIte} with
    \(\texttt{algo = ``nest''}\), \(\texttt{exact = TRUE}\), and CQR.
    \item \textbf{Naive}: the naive ITE conformal method implemented by
    \texttt{cfcausal::conformalIte} with \(\texttt{algo = ``naive''}\) and CQR.
    \item \textbf{DR}: a conformal pseudo-outcome meta-learner using the
    doubly-robust pseudo-outcome.
    \item \textbf{IPW}: a conformal pseudo-outcome meta-learner using the
    inverse-propensity-weighted pseudo-outcome.
    \item \textbf{Split-Pair}: the split-and-pair ITE method.
\end{itemize}
For the \texttt{cfcausal} methods, we use quantile random forests for the
outcome quantile learner and boosting for propensity estimation.  For the meta-learners and
split-pair method, we use gradient-boosting regressors with the default
experiment settings.  The split-pair method uses a training/calibration split
ratio of \(0.5\).

\paragraph{Evaluation metrics.}
On the independent test sample, the true realized ITE
\[
    Y_i(1)-Y_i(0)
\]
is observed by construction.  For a fitted interval
\([L(X_i),U(X_i)]\), the empirical coverage is
\[
    \frac{1}{n_{\mathrm{test}}}
    \sum_{i=1}^{n_{\mathrm{test}}}
    \mathbf{1}\{L(X_i)\le Y_i(1)-Y_i(0)\le U(X_i)\},
\]
and the average interval length is
\[
    \frac{1}{n_{\mathrm{test}}}
    \sum_{i=1}^{n_{\mathrm{test}}}
    \{U(X_i)-L(X_i)\}.
\]
The reported tables average these quantities over the \(20\) Monte Carlo
repetitions.

\subsection{Additional Checkerboard Propensity Stress Test}
\label{app:checkerboard_stress_test}

To probe a regime where the propensity score is much harder to learn, we also
consider an additional checkerboard stress-test design, implemented in the
accompanying code. In this experiment, the covariates are sampled directly from
the uniform distribution on the unit square,
\[
    X\sim \mathrm{Unif}([0,1]^2).
\]
We partition \([0,1]^2\) into a \(10\times 10\) grid of equal-sized cells and
assign treatment according to an alternating checkerboard pattern. Writing the
cell indices as
\[
    i=\min\{\lfloor 10X_1\rfloor,9\},
    \qquad
    j=\min\{\lfloor 10X_2\rfloor,9\},
\]
we set
\[
    e(X)=\mathbb P(T=1\mid X)
    =
    \begin{cases}
        0.95, & i+j \text{ even},\\
        0.05, & i+j \text{ odd}.
    \end{cases}
\]
Thus overlap still holds, but the propensity surface is highly discontinuous and
oscillates at a much finer scale than in the main experiments.

Figure~\ref{fig:grid_checkerboard_design} visualizes this construction. The
left and middle panels show the two complementary sets of hard cells: the
\(Y(0)\)-hard cells are exactly those with high propensity, where control units are
rare, while the \(Y(1)\)-hard cells are exactly those with low propensity, where
treated units are rare. The right panel shows the corresponding checkerboard
propensity surface.

\begin{figure}[H]
\centering
\includegraphics[width=.98\textwidth]{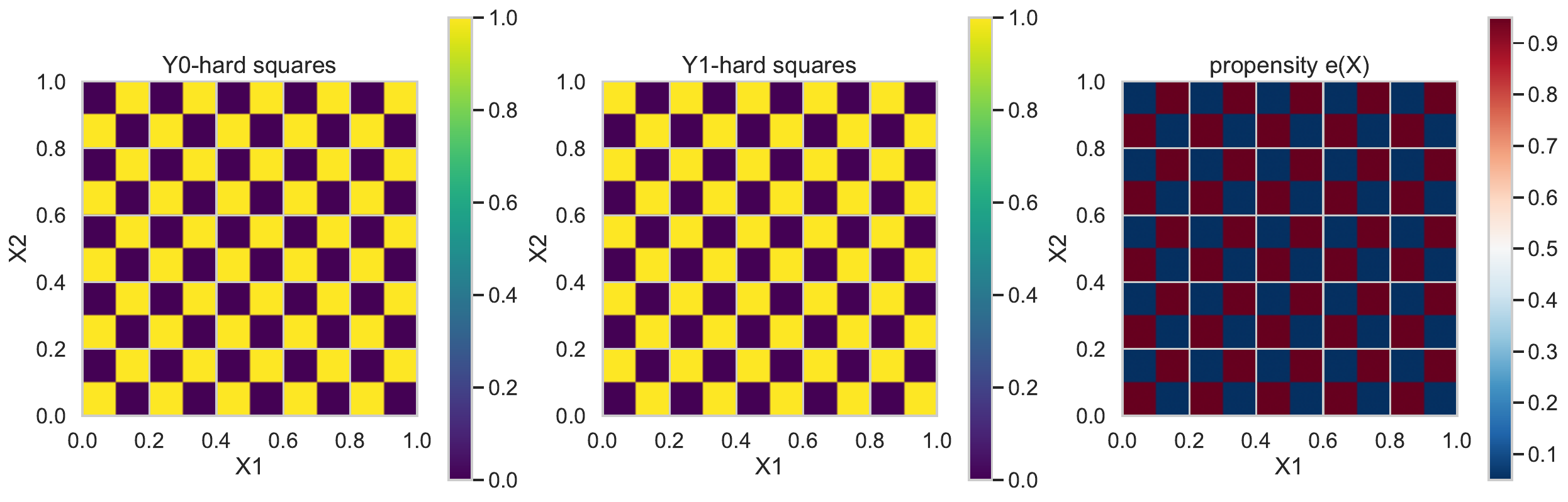}
\caption{Visualization of the checkerboard stress-test design. Left: cells in
which the control potential outcome is harder to learn. Middle: cells in which
the treated potential outcome is harder to learn. Right: the checkerboard
propensity score \(e(X)\), which alternates between \(0.95\) and \(0.05\) over
the \(10\times 10\) grid.}
\label{fig:grid_checkerboard_design}
\end{figure}

The treatment-effect signal is chosen to be the same as in the main simulation
design above.
Let \(R_0(X)\) be the indicator of the high-propensity cells and
\(R_1(X)=1-R_0(X)\). We then define cell-dependent noise scales
\[
    \sigma_0(X)=1+9R_0(X),
    \qquad
    \sigma_1(X)=1+9R_1(X),
\]
so that control outcomes are much noisier in cells where control units are rare and
treated outcomes are much noisier in cells where treated units are rare. The
potential outcomes are generated as
\[
    Y(0)=\sigma_0(X)\varepsilon_0,
    \qquad
    Y(1)=\tau(X)+\sigma_1(X)\varepsilon_1,
\]
where \((\varepsilon_0,\varepsilon_1)\) are standard normal marginally with
correlation \(\rho\in\{-1,-0.5,0\}\). The observed outcome is again
\(Y=TY(1)+(1-T)Y(0)\).

This construction is deliberately adverse for ITE prediction. In high-propensity
cells, there are very few control units precisely where \(Y(0)\) is the hardest to
learn; in low-propensity cells, there are very few treated units precisely where
\(Y(1)\) is the hardest to learn. As a result, finite-sample estimation of the
propensity score and of the two potential-outcome distributions becomes much more
fragile than in the smoother designs above.

We use the same sample sizes, target coverages, and number of Monte Carlo
repetitions as in the main experiments, namely \(n_{\mathrm{train}}=1000\),
\(n_{\mathrm{test}}=10000\), target coverages
\(0.5,0.6,0.7,0.8,0.9\), and \(20\) repetitions. The set of methods is also the
same. For the \texttt{cfcausal} methods, we again use CQR with quantile random
forests for outcome prediction and boosting for propensity estimation, so these
methods must learn the checkerboard propensity from the observed data. For the
DR and IPW meta-learners, the implementation follows the accompanying code; the
remaining tuning parameters are kept fixed across repetitions. Figure~\ref{fig:checkerboard_hard_setting}
in the main text reports the resulting coverage curves.

Table~\ref{tab:grid_checkerboard_ite_all_rhos} reports the corresponding mean
realized coverage and mean interval length under this checkerboard design.

\begin{table}[H]
\centering
\caption{Mean realized coverage and mean interval length under the grid checkerboard DGP. Entries are coverage $\mid$ length.}
\label{tab:grid_checkerboard_ite_all_rhos}
\scriptsize
\resizebox{\textwidth}{!}{%
\begin{tabular}{llcccccc}
\toprule
rho & Target & Nest-Inexact & Nest-Exact & Naive & DR & IPW & Split-Pair \\
\midrule
-1 & 0.50 & 0.078 $\mid$ 2.17 & 0.191 $\mid$ 5.33 & 0.185 $\mid$ 5.18 & 0.093 $\mid$ 2.59 & 0.102 $\mid$ 2.85 & 0.100 $\mid$ 2.81 \\
 & 0.60 & 0.089 $\mid$ 2.48 & 0.219 $\mid$ 6.13 & 0.211 $\mid$ 5.91 & 0.118 $\mid$ 3.29 & 0.126 $\mid$ 3.55 & 0.126 $\mid$ 3.55 \\
 & 0.70 & 0.111 $\mid$ 3.06 & 0.258 $\mid$ 7.29 & 0.246 $\mid$ 6.94 & 0.148 $\mid$ 4.14 & 0.156 $\mid$ 4.37 & 0.166 $\mid$ 4.69 \\
 & 0.80 & 0.143 $\mid$ 4.00 & 0.334 $\mid$ 9.53 & 0.300 $\mid$ 8.51 & 0.195 $\mid$ 5.47 & 0.197 $\mid$ 5.56 & 0.238 $\mid$ 6.79 \\
 & 0.90 & 0.232 $\mid$ 6.57 & 0.799 $\mid$ 29.45 & 0.570 $\mid$ 18.74 & 0.291 $\mid$ 8.31 & 0.281 $\mid$ 8.01 & 0.430 $\mid$ 12.98 \\
\midrule
-0.5 & 0.50 & 0.082 $\mid$ 2.17 & 0.197 $\mid$ 5.25 & 0.198 $\mid$ 5.27 & 0.097 $\mid$ 2.58 & 0.108 $\mid$ 2.88 & 0.103 $\mid$ 2.77 \\
 & 0.60 & 0.100 $\mid$ 2.64 & 0.230 $\mid$ 6.18 & 0.221 $\mid$ 5.92 & 0.124 $\mid$ 3.29 & 0.134 $\mid$ 3.58 & 0.135 $\mid$ 3.61 \\
\bottomrule
\end{tabular}%
}
\end{table}

\subsection{Detailed Results for Target Coverage 0.9}

We report the detailed numerical results for target coverage \(0.9\). The
three tables in this subsection correspond to \(\rho=0.0\), \(\rho=-0.5\), and
\(\rho=-1.0\), respectively.

\begin{table}[H]
\centering
\caption{Mean realized coverage and mean interval length for target coverage 0.9 across propensity settings when $\rho=0.0$. Entries are coverage $\mid$ length.}
\label{tab:cp_ite_target_0p9_rho_0}
\scriptsize
\resizebox{\textwidth}{!}{%
\begin{tabular}{llcccccc}
\toprule
Propensity & Scenario & Nest-Inexact & Nest-Exact & Naive & DR & IPW & Split-Pair \\
\midrule
beta & Homosc. + Ind. & 0.873 $\mid$ 4.76 & 0.998 $\mid$ 9.83 & 0.996 $\mid$ 9.31 & 0.919 $\mid$ 5.52 & 0.943 $\mid$ 6.27 & 0.892 $\mid$ 5.17 \\
 & Heterosc. + Ind. & 0.897 $\mid$ 5.13 & 0.994 $\mid$ 14.30 & 0.993 $\mid$ 9.30 & 0.921 $\mid$ 5.36 & 0.939 $\mid$ 6.52 & 0.886 $\mid$ 6.45 \\
 & Homosc. + Corr. & 0.874 $\mid$ 4.51 & 0.997 $\mid$ 9.04 & 0.995 $\mid$ 8.71 & 0.909 $\mid$ 5.13 & 0.950 $\mid$ 6.83 & 0.895 $\mid$ 5.22 \\
 & Heterosc. + Corr. & 0.896 $\mid$ 4.78 & 0.992 $\mid$ 12.38 & 0.995 $\mid$ 9.46 & 0.908 $\mid$ 4.82 & 0.944 $\mid$ 7.22 & 0.899 $\mid$ 6.97 \\
\midrule
const = 0.1 & Homosc. + Ind. & 0.918 $\mid$ 6.56 & 0.993 $\mid$ 10.02 & 0.995 $\mid$ 10.16 & 0.836 $\mid$ 4.52 & 0.789 $\mid$ 4.57 & 0.864 $\mid$ 5.36 \\
 & Heterosc. + Ind. & 0.908 $\mid$ 8.14 & 0.988 $\mid$ 14.14 & 0.994 $\mid$ 11.64 & 0.846 $\mid$ 4.50 & 0.806 $\mid$ 4.89 & 0.872 $\mid$ 7.04 \\
 & Homosc. + Corr. & 0.901 $\mid$ 6.18 & 0.996 $\mid$ 10.36 & 0.996 $\mid$ 9.57 & 0.829 $\mid$ 4.22 & 0.808 $\mid$ 5.42 & 0.862 $\mid$ 5.15 \\
\bottomrule
\end{tabular}%
}
\end{table}
\begin{table}[H]
\centering
\caption{Mean realized coverage and mean interval length for target coverage 0.9 across propensity settings when $\rho=-0.5$. Entries are coverage $\mid$ length.}
\label{tab:cp_ite_target_0p9_rho_neg05}
\scriptsize
\resizebox{\textwidth}{!}{%
\begin{tabular}{llcccccc}
\toprule
Propensity & Scenario & Nest-Inexact & Nest-Exact & Naive & DR & IPW & Split-Pair \\
\midrule
beta & Homosc. + Ind. & 0.794 $\mid$ 4.72 & 0.989 $\mid$ 9.72 & 0.983 $\mid$ 9.16 & 0.855 $\mid$ 5.43 & 0.894 $\mid$ 6.31 & 0.838 $\mid$ 5.25 \\
 & Heterosc. + Ind. & 0.847 $\mid$ 5.26 & 0.983 $\mid$ 13.13 & 0.983 $\mid$ 9.85 & 0.865 $\mid$ 5.25 & 0.898 $\mid$ 6.46 & 0.857 $\mid$ 6.59 \\
 & Homosc. + Corr. & 0.787 $\mid$ 4.44 & 0.988 $\mid$ 8.99 & 0.984 $\mid$ 8.85 & 0.846 $\mid$ 5.17 & 0.907 $\mid$ 6.86 & 0.831 $\mid$ 5.21 \\
 & Heterosc. + Corr. & 0.839 $\mid$ 4.62 & 0.984 $\mid$ 12.05 & 0.981 $\mid$ 8.47 & 0.854 $\mid$ 4.82 & 0.906 $\mid$ 7.09 & 0.862 $\mid$ 6.52 \\
\midrule
const = 0.1 & Homosc. + Ind. & 0.840 $\mid$ 6.05 & 0.975 $\mid$ 9.62 & 0.984 $\mid$ 10.04 & 0.755 $\mid$ 4.47 & 0.724 $\mid$ 4.59 & 0.835 $\mid$ 5.63 \\
 & Heterosc. + Ind. & 0.860 $\mid$ 6.98 & 0.976 $\mid$ 13.59 & 0.982 $\mid$ 11.35 & 0.785 $\mid$ 4.42 & 0.761 $\mid$ 4.90 & 0.838 $\mid$ 7.05 \\
 & Homosc. + Corr. & 0.821 $\mid$ 5.96 & 0.982 $\mid$ 9.99 & 0.988 $\mid$ 9.87 & 0.751 $\mid$ 4.23 & 0.749 $\mid$ 5.42 & 0.820 $\mid$ 5.42 \\
\bottomrule
\end{tabular}%
}
\end{table}

\begin{table}[H]
\centering
\caption{Mean realized coverage and mean interval length for target coverage 0.9 across propensity settings when $\rho=-1.0$. Entries are coverage $\mid$ length.}
\label{tab:cp_ite_target_0p9_rho_neg1}
\scriptsize
\resizebox{\textwidth}{!}{%
\begin{tabular}{llcccccc}
\toprule
Propensity & Scenario & Nest-Inexact & Nest-Exact & Naive & DR & IPW & Split-Pair \\
\midrule
beta & Homosc. + Ind. & 0.738 $\mid$ 4.71 & 0.977 $\mid$ 9.75 & 0.962 $\mid$ 9.03 & 0.796 $\mid$ 5.38 & 0.845 $\mid$ 6.22 & 0.769 $\mid$ 5.12 \\
 & Heterosc. + Ind. & 0.807 $\mid$ 5.23 & 0.975 $\mid$ 13.10 & 0.966 $\mid$ 9.34 & 0.819 $\mid$ 5.23 & 0.869 $\mid$ 6.42 & 0.829 $\mid$ 6.32 \\
 & Homosc. + Corr. & 0.727 $\mid$ 4.50 & 0.975 $\mid$ 9.33 & 0.954 $\mid$ 8.30 & 0.787 $\mid$ 5.16 & 0.864 $\mid$ 6.83 & 0.776 $\mid$ 5.22 \\
 & Heterosc. + Corr. & 0.796 $\mid$ 6.05 & 0.982 $\mid$ 14.11 & 0.970 $\mid$ 9.50 & 0.817 $\mid$ 4.95 & 0.880 $\mid$ 7.27 & 0.849 $\mid$ 7.17 \\
\midrule
const = 0.1 & Homosc. + Ind. & 0.764 $\mid$ 5.62 & 0.967 $\mid$ 10.24 & 0.972 $\mid$ 10.24 & 0.695 $\mid$ 4.46 & 0.694 $\mid$ 4.83 & 0.769 $\mid$ 5.40 \\
 & Heterosc. + Ind. & 0.845 $\mid$ 7.24 & 0.963 $\mid$ 13.63 & 0.978 $\mid$ 13.84 & 0.737 $\mid$ 4.41 & 0.751 $\mid$ 5.20 & 0.798 $\mid$ 6.03 \\
 & Homosc. + Corr. & 0.770 $\mid$ 5.68 & 0.962 $\mid$ 9.77 & 0.970 $\mid$ 9.35 & 0.682 $\mid$ 4.20 & 0.701 $\mid$ 5.28 & 0.755 $\mid$ 5.23 \\
\bottomrule
\end{tabular}%
}
\end{table}

\subsection{Detailed Results for Target Coverage 0.8}

We report the detailed numerical results for target coverage \(0.8\). The
three tables in this subsection correspond to \(\rho=0.0\), \(\rho=-0.5\), and
\(\rho=-1.0\), respectively.

\begin{table}[H]
\centering
\caption{Mean realized coverage and mean interval length for target coverage 0.8 across propensity settings when $\rho=0.0$. Entries are coverage $\mid$ length.}
\label{tab:cp_ite_target_0p8_rho_0}
\scriptsize
\resizebox{\textwidth}{!}{%
\begin{tabular}{llcccccc}
\toprule
Propensity & Scenario & Nest-Inexact & Nest-Exact & Naive & DR & IPW & Split-Pair \\
\midrule
beta & Homosc. + Ind. & 0.777 $\mid$ 3.75 & 0.988 $\mid$ 7.66 & 0.982 $\mid$ 7.33 & 0.812 $\mid$ 4.04 & 0.862 $\mid$ 4.85 & 0.801 $\mid$ 4.08 \\
 & Heterosc. + Ind. & 0.786 $\mid$ 3.80 & 0.974 $\mid$ 9.27 & 0.975 $\mid$ 7.21 & 0.821 $\mid$ 3.89 & 0.876 $\mid$ 4.92 & 0.791 $\mid$ 3.98 \\
 & Homosc. + Corr. & 0.751 $\mid$ 3.39 & 0.987 $\mid$ 7.36 & 0.980 $\mid$ 6.92 & 0.808 $\mid$ 3.92 & 0.877 $\mid$ 5.40 & 0.793 $\mid$ 4.03 \\
 & Heterosc. + Corr. & 0.786 $\mid$ 3.28 & 0.974 $\mid$ 8.03 & 0.973 $\mid$ 6.38 & 0.806 $\mid$ 3.61 & 0.888 $\mid$ 5.74 & 0.800 $\mid$ 4.10 \\
\midrule
const = 0.1 & Homosc. + Ind. & 0.838 $\mid$ 5.37 & 0.986 $\mid$ 9.23 & 0.982 $\mid$ 8.24 & 0.703 $\mid$ 3.26 & 0.576 $\mid$ 3.07 & 0.791 $\mid$ 4.38 \\
 & Heterosc. + Ind. & 0.804 $\mid$ 5.07 & 0.969 $\mid$ 10.03 & 0.980 $\mid$ 8.91 & 0.704 $\mid$ 3.18 & 0.523 $\mid$ 2.96 & 0.772 $\mid$ 4.16 \\
 & Homosc. + Corr. & 0.811 $\mid$ 4.73 & 0.990 $\mid$ 8.74 & 0.986 $\mid$ 7.96 & 0.702 $\mid$ 3.17 & 0.554 $\mid$ 3.46 & 0.774 $\mid$ 4.17 \\
\bottomrule
\end{tabular}%
}
\end{table}

\begin{table}[H]
\centering
\caption{Mean realized coverage and mean interval length for target coverage 0.8 across propensity settings when $\rho=-0.5$. Entries are coverage $\mid$ length.}
\label{tab:cp_ite_target_0p8_rho_neg05}
\scriptsize
\resizebox{\textwidth}{!}{%
\begin{tabular}{llcccccc}
\toprule
Propensity & Scenario & Nest-Inexact & Nest-Exact & Naive & DR & IPW & Split-Pair \\
\midrule
beta & Homosc. + Ind. & 0.690 $\mid$ 3.74 & 0.967 $\mid$ 7.82 & 0.949 $\mid$ 7.18 & 0.734 $\mid$ 4.07 & 0.785 $\mid$ 4.79 & 0.720 $\mid$ 4.06 \\
 & Heterosc. + Ind. & 0.711 $\mid$ 3.63 & 0.952 $\mid$ 8.94 & 0.947 $\mid$ 7.06 & 0.752 $\mid$ 3.87 & 0.823 $\mid$ 4.89 & 0.752 $\mid$ 3.98 \\
 & Homosc. + Corr. & 0.644 $\mid$ 3.27 & 0.956 $\mid$ 7.21 & 0.950 $\mid$ 6.99 & 0.721 $\mid$ 3.91 & 0.811 $\mid$ 5.35 & 0.714 $\mid$ 4.01 \\
 & Heterosc. + Corr. & 0.722 $\mid$ 3.40 & 0.963 $\mid$ 9.19 & 0.949 $\mid$ 6.66 & 0.732 $\mid$ 3.61 & 0.840 $\mid$ 5.70 & 0.750 $\mid$ 3.85 \\
\midrule
const = 0.1 & Homosc. + Ind. & 0.770 $\mid$ 5.09 & 0.969 $\mid$ 9.05 & 0.961 $\mid$ 8.28 & 0.617 $\mid$ 3.25 & 0.514 $\mid$ 3.07 & 0.748 $\mid$ 4.60 \\
 & Heterosc. + Ind. & 0.766 $\mid$ 5.69 & 0.952 $\mid$ 10.76 & 0.948 $\mid$ 8.01 & 0.644 $\mid$ 3.17 & 0.474 $\mid$ 2.98 & 0.750 $\mid$ 4.30 \\
 & Homosc. + Corr. & 0.719 $\mid$ 4.49 & 0.963 $\mid$ 8.74 & 0.949 $\mid$ 7.28 & 0.624 $\mid$ 3.21 & 0.511 $\mid$ 3.47 & 0.696 $\mid$ 4.14 \\
\bottomrule
\end{tabular}%
}
\end{table}

\begin{table}[H]
\centering
\caption{Mean realized coverage and mean interval length for target coverage 0.8 across propensity settings when $\rho=-1.0$. Entries are coverage $\mid$ length.}
\label{tab:cp_ite_target_0p8_rho_neg1}
\scriptsize
\resizebox{\textwidth}{!}{%
\begin{tabular}{llcccccc}
\toprule
Propensity & Scenario & Nest-Inexact & Nest-Exact & Naive & DR & IPW & Split-Pair \\
\midrule
beta & Homosc. + Ind. & 0.610 $\mid$ 3.58 & 0.943 $\mid$ 8.00 & 0.909 $\mid$ 7.09 & 0.656 $\mid$ 3.96 & 0.719 $\mid$ 4.71 & 0.639 $\mid$ 3.90 \\
 & Heterosc. + Ind. & 0.661 $\mid$ 3.62 & 0.932 $\mid$ 8.64 & 0.924 $\mid$ 7.19 & 0.696 $\mid$ 3.87 & 0.783 $\mid$ 4.89 & 0.715 $\mid$ 3.88 \\
 & Homosc. + Corr. & 0.598 $\mid$ 3.41 & 0.923 $\mid$ 7.25 & 0.915 $\mid$ 7.11 & 0.659 $\mid$ 3.93 & 0.764 $\mid$ 5.39 & 0.659 $\mid$ 4.05 \\
 & Heterosc. + Corr. & 0.651 $\mid$ 3.31 & 0.942 $\mid$ 8.96 & 0.918 $\mid$ 6.59 & 0.678 $\mid$ 3.69 & 0.812 $\mid$ 5.83 & 0.725 $\mid$ 3.99 \\
\midrule
const = 0.1 & Homosc. + Ind. & 0.690 $\mid$ 4.82 & 0.944 $\mid$ 9.10 & 0.923 $\mid$ 8.08 & 0.561 $\mid$ 3.26 & 0.501 $\mid$ 3.28 & 0.641 $\mid$ 4.11 \\
 & Heterosc. + Ind. & 0.747 $\mid$ 5.15 & 0.936 $\mid$ 10.46 & 0.941 $\mid$ 9.15 & 0.603 $\mid$ 3.19 & 0.466 $\mid$ 3.15 & 0.711 $\mid$ 4.25 \\
 & Homosc. + Corr. & 0.707 $\mid$ 4.96 & 0.953 $\mid$ 9.08 & 0.910 $\mid$ 7.25 & 0.542 $\mid$ 3.08 & 0.470 $\mid$ 3.41 & 0.632 $\mid$ 4.01 \\
\bottomrule
\end{tabular}%
}
\end{table}

\subsection{Detailed Results for Target Coverage 0.7}

We report the detailed numerical results for target coverage \(0.7\). The
three tables in this subsection correspond to \(\rho=0.0\), \(\rho=-0.5\), and
\(\rho=-1.0\), respectively.

\begin{table}[H]
\centering
\caption{Mean realized coverage and mean interval length for target coverage 0.7 across propensity settings when $\rho=0.0$. Entries are coverage $\mid$ length.}
\label{tab:cp_ite_target_0p7_rho_0}
\scriptsize
\resizebox{\textwidth}{!}{%
\begin{tabular}{llcccccc}
\toprule
Propensity & Scenario & Nest-Inexact & Nest-Exact & Naive & DR & IPW & Split-Pair \\
\midrule
beta & Homosc. + Ind. & 0.685 $\mid$ 3.00 & 0.968 $\mid$ 6.43 & 0.961 $\mid$ 6.17 & 0.717 $\mid$ 3.24 & 0.772 $\mid$ 3.95 & 0.698 $\mid$ 3.27 \\
 & Heterosc. + Ind. & 0.707 $\mid$ 2.90 & 0.941 $\mid$ 6.45 & 0.949 $\mid$ 5.64 & 0.726 $\mid$ 3.05 & 0.812 $\mid$ 4.08 & 0.699 $\mid$ 2.76 \\
 & Homosc. + Corr. & 0.666 $\mid$ 2.82 & 0.965 $\mid$ 6.22 & 0.955 $\mid$ 5.92 & 0.706 $\mid$ 3.15 & 0.807 $\mid$ 4.49 & 0.698 $\mid$ 3.27 \\
 & Heterosc. + Corr. & 0.718 $\mid$ 2.77 & 0.948 $\mid$ 6.40 & 0.952 $\mid$ 5.56 & 0.716 $\mid$ 2.95 & 0.836 $\mid$ 4.90 & 0.701 $\mid$ 2.68 \\
\midrule
const = 0.1 & Homosc. + Ind. & 0.712 $\mid$ 3.88 & 0.975 $\mid$ 7.90 & 0.967 $\mid$ 7.22 & 0.587 $\mid$ 2.53 & 0.463 $\mid$ 2.40 & 0.674 $\mid$ 3.38 \\
 & Heterosc. + Ind. & 0.734 $\mid$ 4.35 & 0.951 $\mid$ 8.79 & 0.951 $\mid$ 6.81 & 0.594 $\mid$ 2.45 & 0.396 $\mid$ 2.18 & 0.687 $\mid$ 2.94 \\
 & Homosc. + Corr. & 0.734 $\mid$ 3.79 & 0.974 $\mid$ 7.34 & 0.968 $\mid$ 6.68 & 0.591 $\mid$ 2.50 & 0.411 $\mid$ 2.57 & 0.674 $\mid$ 3.38 \\
\bottomrule
\end{tabular}%
}
\end{table}

\begin{table}[H]
\centering
\caption{Mean realized coverage and mean interval length for target coverage 0.7 across propensity settings when $\rho=-0.5$. Entries are coverage $\mid$ length.}
\label{tab:cp_ite_target_0p7_rho_neg05}
\scriptsize
\resizebox{\textwidth}{!}{%
\begin{tabular}{llcccccc}
\toprule
Propensity & Scenario & Nest-Inexact & Nest-Exact & Naive & DR & IPW & Split-Pair \\
\midrule
beta & Homosc. + Ind. & 0.589 $\mid$ 2.97 & 0.923 $\mid$ 6.37 & 0.912 $\mid$ 6.13 & 0.630 $\mid$ 3.26 & 0.683 $\mid$ 3.88 & 0.613 $\mid$ 3.25 \\
 & Heterosc. + Ind. & 0.653 $\mid$ 2.95 & 0.916 $\mid$ 6.92 & 0.905 $\mid$ 5.58 & 0.646 $\mid$ 3.01 & 0.755 $\mid$ 4.06 & 0.646 $\mid$ 2.65 \\
 & Homosc. + Corr. & 0.572 $\mid$ 2.81 & 0.920 $\mid$ 6.20 & 0.908 $\mid$ 6.03 & 0.619 $\mid$ 3.14 & 0.736 $\mid$ 4.47 & 0.607 $\mid$ 3.19 \\
 & Heterosc. + Corr. & 0.592 $\mid$ 2.61 & 0.914 $\mid$ 6.39 & 0.906 $\mid$ 5.47 & 0.632 $\mid$ 2.93 & 0.789 $\mid$ 4.94 & 0.653 $\mid$ 2.63 \\
\midrule
const = 0.1 & Homosc. + Ind. & 0.668 $\mid$ 4.04 & 0.949 $\mid$ 8.00 & 0.924 $\mid$ 7.06 & 0.515 $\mid$ 2.56 & 0.412 $\mid$ 2.41 & 0.613 $\mid$ 3.43 \\
 & Heterosc. + Ind. & 0.671 $\mid$ 4.15 & 0.912 $\mid$ 7.93 & 0.915 $\mid$ 7.17 & 0.533 $\mid$ 2.42 & 0.350 $\mid$ 2.16 & 0.642 $\mid$ 2.82 \\
 & Homosc. + Corr. & 0.643 $\mid$ 3.66 & 0.942 $\mid$ 7.54 & 0.920 $\mid$ 6.46 & 0.513 $\mid$ 2.52 & 0.375 $\mid$ 2.57 & 0.603 $\mid$ 3.37 \\
\bottomrule
\end{tabular}%
}
\end{table}

\begin{table}[H]
\centering
\caption{Mean realized coverage and mean interval length for target coverage 0.7 across propensity settings when $\rho=-1.0$. Entries are coverage $\mid$ length.}
\label{tab:cp_ite_target_0p7_rho_neg1}
\scriptsize
\resizebox{\textwidth}{!}{%
\begin{tabular}{llcccccc}
\toprule
Propensity & Scenario & Nest-Inexact & Nest-Exact & Naive & DR & IPW & Split-Pair \\
\midrule
beta & Homosc. + Ind. & 0.522 $\mid$ 2.93 & 0.884 $\mid$ 6.53 & 0.872 $\mid$ 6.27 & 0.551 $\mid$ 3.14 & 0.616 $\mid$ 3.80 & 0.535 $\mid$ 3.12 \\
 & Heterosc. + Ind. & 0.594 $\mid$ 2.86 & 0.886 $\mid$ 6.64 & 0.864 $\mid$ 5.44 & 0.582 $\mid$ 2.98 & 0.707 $\mid$ 4.02 & 0.601 $\mid$ 2.53 \\
 & Homosc. + Corr. & 0.506 $\mid$ 2.77 & 0.875 $\mid$ 6.24 & 0.856 $\mid$ 5.98 & 0.554 $\mid$ 3.14 & 0.684 $\mid$ 4.51 & 0.554 $\mid$ 3.24 \\
 & Heterosc. + Corr. & 0.559 $\mid$ 2.66 & 0.894 $\mid$ 6.75 & 0.859 $\mid$ 5.25 & 0.570 $\mid$ 3.00 & 0.754 $\mid$ 4.97 & 0.622 $\mid$ 2.66 \\
\midrule
const = 0.1 & Homosc. + Ind. & 0.583 $\mid$ 3.77 & 0.914 $\mid$ 7.90 & 0.891 $\mid$ 7.22 & 0.460 $\mid$ 2.58 & 0.393 $\mid$ 2.52 & 0.557 $\mid$ 3.42 \\
 & Heterosc. + Ind. & 0.572 $\mid$ 3.60 & 0.907 $\mid$ 9.03 & 0.885 $\mid$ 7.24 & 0.491 $\mid$ 2.43 & 0.338 $\mid$ 2.30 & 0.620 $\mid$ 3.05 \\
 & Homosc. + Corr. & 0.583 $\mid$ 3.59 & 0.910 $\mid$ 7.65 & 0.864 $\mid$ 6.29 & 0.449 $\mid$ 2.47 & 0.354 $\mid$ 2.60 & 0.533 $\mid$ 3.21 \\
\bottomrule
\end{tabular}%
}
\end{table}

\subsection{Detailed Results for Target Coverage 0.6}

We report the detailed numerical results for target coverage \(0.6\). The
three tables in this subsection correspond to \(\rho=0.0\), \(\rho=-0.5\), and
\(\rho=-1.0\), respectively.
\begin{table}[H]
\centering
\caption{Mean realized coverage and mean interval length for target coverage 0.6 across propensity settings when $\rho=0.0$. Entries are coverage $\mid$ length.}
\label{tab:cp_ite_target_0p6_rho_0}
\scriptsize
\resizebox{\textwidth}{!}{%
\begin{tabular}{llcccccc}
\toprule
Propensity & Scenario & Nest-Inexact & Nest-Exact & Naive & DR & IPW & Split-Pair \\
\midrule
beta & Homosc. + Ind. & 0.585 $\mid$ 2.40 & 0.944 $\mid$ 5.69 & 0.930 $\mid$ 5.35 & 0.615 $\mid$ 2.62 & 0.684 $\mid$ 3.30 & 0.593 $\mid$ 2.63 \\
 & Heterosc. + Ind. & 0.617 $\mid$ 2.22 & 0.912 $\mid$ 5.34 & 0.916 $\mid$ 4.78 & 0.623 $\mid$ 2.44 & 0.735 $\mid$ 3.37 & 0.598 $\mid$ 1.90 \\
 & Homosc. + Corr. & 0.564 $\mid$ 2.27 & 0.945 $\mid$ 5.62 & 0.942 $\mid$ 5.57 & 0.602 $\mid$ 2.53 & 0.719 $\mid$ 3.73 & 0.601 $\mid$ 2.66 \\
 & Heterosc. + Corr. & 0.601 $\mid$ 2.22 & 0.909 $\mid$ 5.06 & 0.919 $\mid$ 4.85 & 0.607 $\mid$ 2.39 & 0.775 $\mid$ 4.20 & 0.597 $\mid$ 1.80 \\
\midrule
const = 0.1 & Homosc. + Ind. & 0.624 $\mid$ 3.12 & 0.953 $\mid$ 7.13 & 0.934 $\mid$ 6.16 & 0.496 $\mid$ 2.04 & 0.378 $\mid$ 1.91 & 0.579 $\mid$ 2.72 \\
 & Heterosc. + Ind. & 0.600 $\mid$ 3.18 & 0.925 $\mid$ 7.41 & 0.920 $\mid$ 6.00 & 0.494 $\mid$ 1.97 & 0.318 $\mid$ 1.71 & 0.587 $\mid$ 2.11 \\
 & Homosc. + Corr. & 0.592 $\mid$ 2.68 & 0.945 $\mid$ 6.37 & 0.938 $\mid$ 5.71 & 0.498 $\mid$ 2.03 & 0.328 $\mid$ 2.04 & 0.591 $\mid$ 2.82 \\
\bottomrule
\end{tabular}%
}
\end{table}

\begin{table}[H]
\centering
\caption{Mean realized coverage and mean interval length for target coverage 0.6 across propensity settings when $\rho=-0.5$. Entries are coverage $\mid$ length.}
\label{tab:cp_ite_target_0p6_rho_neg05}
\scriptsize
\resizebox{\textwidth}{!}{%
\begin{tabular}{llcccccc}
\toprule
Propensity & Scenario & Nest-Inexact & Nest-Exact & Naive & DR & IPW & Split-Pair \\
\midrule
beta & Homosc. + Ind. & 0.501 $\mid$ 2.41 & 0.885 $\mid$ 5.65 & 0.872 $\mid$ 5.42 & 0.526 $\mid$ 2.59 & 0.590 $\mid$ 3.21 & 0.518 $\mid$ 2.63 \\
 & Heterosc. + Ind. & 0.548 $\mid$ 2.22 & 0.876 $\mid$ 5.47 & 0.860 $\mid$ 4.61 & 0.548 $\mid$ 2.43 & 0.672 $\mid$ 3.36 & 0.543 $\mid$ 1.82 \\
 & Homosc. + Corr. & 0.494 $\mid$ 2.35 & 0.886 $\mid$ 5.64 & 0.875 $\mid$ 5.47 & 0.511 $\mid$ 2.49 & 0.651 $\mid$ 3.76 & 0.518 $\mid$ 2.62 \\
 & Heterosc. + Corr. & 0.543 $\mid$ 2.22 & 0.879 $\mid$ 5.39 & 0.872 $\mid$ 4.86 & 0.530 $\mid$ 2.37 & 0.726 $\mid$ 4.23 & 0.548 $\mid$ 1.75 \\
\midrule
const = 0.1 & Homosc. + Ind. & 0.580 $\mid$ 3.32 & 0.925 $\mid$ 7.49 & 0.877 $\mid$ 6.01 & 0.420 $\mid$ 2.02 & 0.330 $\mid$ 1.90 & 0.520 $\mid$ 2.77 \\
 & Heterosc. + Ind. & 0.559 $\mid$ 3.21 & 0.891 $\mid$ 7.67 & 0.878 $\mid$ 5.87 & 0.433 $\mid$ 1.92 & 0.278 $\mid$ 1.71 & 0.534 $\mid$ 1.96 \\
 & Homosc. + Corr. & 0.547 $\mid$ 2.86 & 0.916 $\mid$ 6.72 & 0.884 $\mid$ 5.75 & 0.422 $\mid$ 2.02 & 0.296 $\mid$ 2.04 & 0.511 $\mid$ 2.75 \\
\bottomrule
\end{tabular}%
}
\end{table}

\begin{table}[H]
\centering
\caption{Mean realized coverage and mean interval length for target coverage 0.6 across propensity settings when $\rho=-1.0$. Entries are coverage $\mid$ length.}
\label{tab:cp_ite_target_0p6_rho_neg1}
\scriptsize
\resizebox{\textwidth}{!}{%
\begin{tabular}{llcccccc}
\toprule
Propensity & Scenario & Nest-Inexact & Nest-Exact & Naive & DR & IPW & Split-Pair \\
\midrule
beta & Homosc. + Ind. & 0.429 $\mid$ 2.32 & 0.827 $\mid$ 5.58 & 0.806 $\mid$ 5.32 & 0.459 $\mid$ 2.53 & 0.525 $\mid$ 3.13 & 0.442 $\mid$ 2.49 \\
 & Heterosc. + Ind. & 0.512 $\mid$ 2.25 & 0.841 $\mid$ 5.43 & 0.804 $\mid$ 4.48 & 0.489 $\mid$ 2.42 & 0.628 $\mid$ 3.34 & 0.500 $\mid$ 1.75 \\
 & Homosc. + Corr. & 0.434 $\mid$ 2.33 & 0.825 $\mid$ 5.52 & 0.805 $\mid$ 5.27 & 0.463 $\mid$ 2.54 & 0.593 $\mid$ 3.75 & 0.467 $\mid$ 2.65 \\
 & Heterosc. + Corr. & 0.483 $\mid$ 2.21 & 0.833 $\mid$ 5.07 & 0.821 $\mid$ 4.67 & 0.479 $\mid$ 2.41 & 0.694 $\mid$ 4.24 & 0.517 $\mid$ 1.78 \\
\midrule
const = 0.1 & Homosc. + Ind. & 0.517 $\mid$ 3.20 & 0.856 $\mid$ 6.79 & 0.835 $\mid$ 6.11 & 0.379 $\mid$ 2.06 & 0.311 $\mid$ 1.98 & 0.461 $\mid$ 2.73 \\
 & Heterosc. + Ind. & 0.535 $\mid$ 3.26 & 0.865 $\mid$ 8.14 & 0.834 $\mid$ 5.77 & 0.403 $\mid$ 1.94 & 0.267 $\mid$ 1.81 & 0.513 $\mid$ 2.13 \\
 & Homosc. + Corr. & 0.460 $\mid$ 2.67 & 0.849 $\mid$ 6.30 & 0.818 $\mid$ 5.59 & 0.369 $\mid$ 1.99 & 0.273 $\mid$ 2.02 & 0.467 $\mid$ 2.75 \\
\bottomrule
\end{tabular}%
}
\end{table}

\subsection{Detailed Results for Target Coverage 0.5}

We report the detailed numerical results for target coverage \(0.5\). The
three tables in this subsection correspond to \(\rho=0.0\), \(\rho=-0.5\), and
\(\rho=-1.0\), respectively.

\begin{table}[H]
\centering
\caption{Mean realized coverage and mean interval length for target coverage 0.5 across propensity settings when $\rho=0.0$. Entries are coverage $\mid$ length.}
\label{tab:cp_ite_target_0p5_rho_0}
\scriptsize
\resizebox{\textwidth}{!}{%
\begin{tabular}{llcccccc}
\toprule
Propensity & Scenario & Nest-Inexact & Nest-Exact & Naive & DR & IPW & Split-Pair \\
\midrule
beta & Homosc. + Ind. & 0.530 $\mid$ 2.12 & 0.913 $\mid$ 5.04 & 0.898 $\mid$ 4.79 & 0.517 $\mid$ 2.11 & 0.578 $\mid$ 2.67 & 0.494 $\mid$ 2.11 \\
 & Heterosc. + Ind. & 0.527 $\mid$ 1.76 & 0.878 $\mid$ 4.54 & 0.868 $\mid$ 3.93 & 0.520 $\mid$ 1.95 & 0.650 $\mid$ 2.80 & 0.500 $\mid$ 1.34 \\
 & Homosc. + Corr. & 0.504 $\mid$ 1.98 & 0.917 $\mid$ 5.07 & 0.904 $\mid$ 4.88 & 0.503 $\mid$ 2.02 & 0.627 $\mid$ 3.09 & 0.510 $\mid$ 2.18 \\
 & Heterosc. + Corr. & 0.539 $\mid$ 1.87 & 0.878 $\mid$ 4.28 & 0.863 $\mid$ 3.81 & 0.502 $\mid$ 1.91 & 0.705 $\mid$ 3.61 & 0.507 $\mid$ 1.27 \\
\midrule
const = 0.1 & Homosc. + Ind. & 0.553 $\mid$ 2.75 & 0.908 $\mid$ 5.97 & 0.896 $\mid$ 5.27 & 0.406 $\mid$ 1.62 & 0.296 $\mid$ 1.47 & 0.472 $\mid$ 2.13 \\
 & Heterosc. + Ind. & 0.571 $\mid$ 2.94 & 0.887 $\mid$ 6.17 & 0.878 $\mid$ 4.88 & 0.393 $\mid$ 1.52 & 0.249 $\mid$ 1.34 & 0.467 $\mid$ 1.42 \\
 & Homosc. + Corr. & 0.483 $\mid$ 2.06 & 0.909 $\mid$ 5.40 & 0.894 $\mid$ 4.90 & 0.409 $\mid$ 1.62 & 0.262 $\mid$ 1.63 & 0.491 $\mid$ 2.23 \\
\bottomrule
\end{tabular}%
}
\end{table}

\begin{table}[H]
\centering
\caption{Mean realized coverage and mean interval length for target coverage 0.5 across propensity settings when $\rho=-0.5$. Entries are coverage $\mid$ length.}
\label{tab:cp_ite_target_0p5_rho_neg05}
\scriptsize
\resizebox{\textwidth}{!}{%
\begin{tabular}{llcccccc}
\toprule
Propensity & Scenario & Nest-Inexact & Nest-Exact & Naive & DR & IPW & Split-Pair \\
\midrule
beta & Homosc. + Ind. & 0.439 $\mid$ 2.07 & 0.834 $\mid$ 4.93 & 0.822 $\mid$ 4.78 & 0.433 $\mid$ 2.07 & 0.501 $\mid$ 2.65 & 0.426 $\mid$ 2.10 \\
 & Heterosc. + Ind. & 0.475 $\mid$ 1.80 & 0.823 $\mid$ 4.36 & 0.809 $\mid$ 3.91 & 0.453 $\mid$ 1.94 & 0.596 $\mid$ 2.81 & 0.449 $\mid$ 1.30 \\
 & Homosc. + Corr. & 0.414 $\mid$ 1.92 & 0.837 $\mid$ 4.92 & 0.821 $\mid$ 4.77 & 0.417 $\mid$ 1.96 & 0.560 $\mid$ 3.11 & 0.421 $\mid$ 2.07 \\
 & Heterosc. + Corr. & 0.445 $\mid$ 1.83 & 0.824 $\mid$ 4.27 & 0.798 $\mid$ 3.79 & 0.432 $\mid$ 1.93 & 0.656 $\mid$ 3.63 & 0.453 $\mid$ 1.21 \\
\midrule
const = 0.1 & Homosc. + Ind. & 0.481 $\mid$ 2.63 & 0.881 $\mid$ 6.46 & 0.838 $\mid$ 5.34 & 0.346 $\mid$ 1.63 & 0.264 $\mid$ 1.50 & 0.421 $\mid$ 2.17 \\
 & Heterosc. + Ind. & 0.499 $\mid$ 2.74 & 0.861 $\mid$ 6.31 & 0.838 $\mid$ 5.04 & 0.345 $\mid$ 1.49 & 0.220 $\mid$ 1.35 & 0.443 $\mid$ 1.43 \\
 & Homosc. + Corr. & 0.467 $\mid$ 2.34 & 0.873 $\mid$ 5.78 & 0.822 $\mid$ 4.91 & 0.342 $\mid$ 1.60 & 0.232 $\mid$ 1.61 & 0.404 $\mid$ 2.09 \\
\bottomrule
\end{tabular}%
}
\end{table}

\begin{table}[H]
\centering
\caption{Mean realized coverage and mean interval length for target coverage 0.5 across propensity settings when $\rho=-1.0$. Entries are coverage $\mid$ length.}
\label{tab:cp_ite_target_0p5_rho_neg1}
\scriptsize
\resizebox{\textwidth}{!}{%
\begin{tabular}{llcccccc}
\toprule
Propensity & Scenario & Nest-Inexact & Nest-Exact & Naive & DR & IPW & Split-Pair \\
\midrule
beta & Homosc. + Ind. & 0.370 $\mid$ 1.98 & 0.773 $\mid$ 4.93 & 0.763 $\mid$ 4.83 & 0.377 $\mid$ 2.03 & 0.439 $\mid$ 2.56 & 0.354 $\mid$ 1.96 \\
 & Heterosc. + Ind. & 0.420 $\mid$ 1.76 & 0.780 $\mid$ 4.24 & 0.769 $\mid$ 3.95 & 0.395 $\mid$ 1.92 & 0.548 $\mid$ 2.79 & 0.410 $\mid$ 1.25 \\
 & Homosc. + Corr. & 0.367 $\mid$ 1.93 & 0.776 $\mid$ 4.94 & 0.762 $\mid$ 4.81 & 0.378 $\mid$ 2.02 & 0.503 $\mid$ 3.07 & 0.383 $\mid$ 2.12 \\
 & Heterosc. + Corr. & 0.410 $\mid$ 1.85 & 0.773 $\mid$ 4.03 & 0.776 $\mid$ 4.08 & 0.386 $\mid$ 1.93 & 0.619 $\mid$ 3.60 & 0.427 $\mid$ 1.24 \\
\midrule
const = 0.1 & Homosc. + Ind. & 0.468 $\mid$ 2.84 & 0.781 $\mid$ 5.60 & 0.792 $\mid$ 5.45 & 0.310 $\mid$ 1.66 & 0.249 $\mid$ 1.57 & 0.366 $\mid$ 2.11 \\
 & Heterosc. + Ind. & 0.497 $\mid$ 2.93 & 0.819 $\mid$ 6.03 & 0.792 $\mid$ 4.82 & 0.324 $\mid$ 1.52 & 0.207 $\mid$ 1.39 & 0.418 $\mid$ 1.53 \\
 & Homosc. + Corr. & 0.431 $\mid$ 2.43 & 0.803 $\mid$ 5.55 & 0.773 $\mid$ 5.01 & 0.296 $\mid$ 1.57 & 0.211 $\mid$ 1.58 & 0.375 $\mid$ 2.17 \\
\bottomrule
\end{tabular}%
}
\end{table}

%% file: checklist.tex
\section*{NeurIPS Paper Checklist}

\begin{enumerate}

\item {\bf Claims}
    \item[] Question: Do the main claims made in the abstract and introduction accurately reflect the paper's contributions and scope?
    \item[] Answer: \answerYes{}
    \item[] Justification: 
    The main claims in the abstract and introduction accurately reflect the contributions and scope of the paper. 
    The paper clearly states its impossibility results for distribution-free predictive inference of ITEs under standard causal assumptions, its positive result under additional structural assumptions, and its empirical illustrations in Sections 2--4.
    \item[] Guidelines:
    \begin{itemize}
        \item The answer \answerNA{} means that the abstract and introduction do not include the claims made in the paper.
        \item The abstract and/or introduction should clearly state the claims made, including the contributions made in the paper and important assumptions and limitations. A \answerNo{} or \answerNA{} answer to this question will not be perceived well by the reviewers. 
        \item The claims made should match theoretical and experimental results, and reflect how much the results can be expected to generalize to other settings. 
        \item It is fine to include aspirational goals as motivation as long as it is clear that these goals are not attained by the paper. 
    \end{itemize}

\item {\bf Limitations}
    \item[] Question: Does the paper discuss the limitations of the work performed by the authors?
    \item[] Answer: \answerYes{}
    \item[] Justification: 
    The paper discusses the scope and limitations of the results. 
    In particular, the infinite-expected-length conclusion for distribution-free prediction of realized ITEs requires no regularity assumption on the covariates and holds even in their absence, provided coverage holds uniformly over the corresponding strong-overlap class. The stronger pointwise inclusion-probability results use an absolutely continuous covariate coordinate. The positive Split-Pair result requires additional structural assumptions, including covariate-independent residuals and conditional independence between potential outcomes.
    \item[] Guidelines:
    \begin{itemize}
        \item The answer \answerNA{} means that the paper has no limitation while the answer \answerNo{} means that the paper has limitations, but those are not discussed in the paper. 
        \item The authors are encouraged to create a separate ``Limitations'' section in their paper.
        \item The paper should point out any strong assumptions and how robust the results are to violations of these assumptions (e.g., independence assumptions, noiseless settings, model well-specification, asymptotic approximations only holding locally). The authors should reflect on how these assumptions might be violated in practice and what the implications would be.
        \item The authors should reflect on the scope of the claims made, e.g., if the approach was only tested on a few datasets or with a few runs. In general, empirical results often depend on implicit assumptions, which should be articulated.
        \item The authors should reflect on the factors that influence the performance of the approach. For example, a facial recognition algorithm may perform poorly when image resolution is low or images are taken in low lighting. Or a speech-to-text system might not be used reliably to provide closed captions for online lectures because it fails to handle technical jargon.
        \item The authors should discuss the computational efficiency of the proposed algorithms and how they scale with dataset size.
        \item If applicable, the authors should discuss possible limitations of their approach to address problems of privacy and fairness.
        \item While the authors might fear that complete honesty about limitations might be used by reviewers as grounds for rejection, a worse outcome might be that reviewers discover limitations that aren't acknowledged in the paper. The authors should use their best judgment and recognize that individual actions in favor of transparency play an important role in developing norms that preserve the integrity of the community. Reviewers will be specifically instructed to not penalize honesty concerning limitations.
    \end{itemize}

\item {\bf Theory assumptions and proofs}
    \item[] Question: For each theoretical result, does the paper provide the full set of assumptions and a complete (and correct) proof?
    \item[] Answer: \answerYes{}
    \item[] Justification: 
    All theoretical results are stated with their assumptions in Section 2 and Appendix A. 
    Complete proofs of the impossibility results, supporting lemmas, and the positive Split-Pair result are provided in Appendix A, with the main proof ingredients and external results properly referenced.
    \item[] Guidelines:
    \begin{itemize}
        \item The answer \answerNA{} means that the paper does not include theoretical results. 
        \item All the theorems, formulas, and proofs in the paper should be numbered and cross-referenced.
        \item All assumptions should be clearly stated or referenced in the statement of any theorems.
        \item The proofs can either appear in the main paper or the supplemental material, but if they appear in the supplemental material, the authors are encouraged to provide a short proof sketch to provide intuition. 
        \item Inversely, any informal proof provided in the core of the paper should be complemented by formal proofs provided in appendix or supplemental material.
        \item Theorems and Lemmas that the proof relies upon should be properly referenced. 
    \end{itemize}

\item {\bf Experimental result reproducibility}
\item[] Question: Does the paper fully disclose all the information needed to reproduce the main experimental results of the paper to the extent that it affects the main claims and/or conclusions of the paper (regardless of whether the code and data are provided or not)?
\item[] Answer: \answerYes{}
\item[] Justification: 
The paper provides the information needed to reproduce the main simulation results. 
The data-generating processes, propensity score regimes, error dependence structures, compared methods, evaluation metrics, and additional implementation details are described in Section 4 and Appendix C.
\item[] Guidelines:
\begin{itemize}
    \item The answer \answerNA{} means that the paper does not include experiments.
    \item If the paper includes experiments, a \answerNo{} answer to this question will not be perceived well by the reviewers: Making the paper reproducible is important, regardless of whether the code and data are provided or not.
    \item If the contribution is a dataset and\slash or model, the authors should describe the steps taken to make their results reproducible or verifiable. 
    \item Depending on the contribution, reproducibility can be accomplished in various ways. For example, if the contribution is a novel architecture, describing the architecture fully might suffice, or if the contribution is a specific model and empirical evaluation, it may be necessary to either make it possible for others to replicate the model with the same dataset, or provide access to the model. In general. releasing code and data is often one good way to accomplish this, but reproducibility can also be provided via detailed instructions for how to replicate the results, access to a hosted model (e.g., in the case of a large language model), releasing of a model checkpoint, or other means that are appropriate to the research performed.
    \item While NeurIPS does not require releasing code, the conference does require all submissions to provide some reasonable avenue for reproducibility, which may depend on the nature of the contribution. For example
    \begin{enumerate}
        \item If the contribution is primarily a new algorithm, the paper should make it clear how to reproduce that algorithm.
        \item If the contribution is primarily a new model architecture, the paper should describe the architecture clearly and fully.
        \item If the contribution is a new model (e.g., a large language model), then there should either be a way to access this model for reproducing the results or a way to reproduce the model (e.g., with an open-source dataset or instructions for how to construct the dataset).
        \item We recognize that reproducibility may be tricky in some cases, in which case authors are welcome to describe the particular way they provide for reproducibility. In the case of closed-source models, it may be that access to the model is limited in some way (e.g., to registered users), but it should be possible for other researchers to have some path to reproducing or verifying the results.
    \end{enumerate}
\end{itemize}

\item {\bf Open access to data and code}
    \item[] Question: Does the paper provide open access to the data and code, with sufficient instructions to faithfully reproduce the main experimental results, as described in supplemental material?
    \item[] Answer: \answerYes{}
    \item[] Justification: 
    The paper provides anonymized simulation code and reproduction instructions in the supplemental material. 
    All experiments are based on synthetic data generated from the data-generating processes described in Section 4 and Appendix C, and the code will be de-anonymized and publicly released upon acceptance.
    \item[] Guidelines:
    \begin{itemize}
        \item The answer \answerNA{} means that paper does not include experiments requiring code.
        \item Please see the NeurIPS code and data submission guidelines (\url{https://neurips.cc/public/guides/CodeSubmissionPolicy}) for more details.
        \item While we encourage the release of code and data, we understand that this might not be possible, so \answerNo{} is an acceptable answer. Papers cannot be rejected simply for not including code, unless this is central to the contribution (e.g., for a new open-source benchmark).
        \item The instructions should contain the exact command and environment needed to run to reproduce the results. See the NeurIPS code and data submission guidelines (\url{https://neurips.cc/public/guides/CodeSubmissionPolicy}) for more details.
        \item The authors should provide instructions on data access and preparation, including how to access the raw data, preprocessed data, intermediate data, and generated data, etc.
        \item The authors should provide scripts to reproduce all experimental results for the new proposed method and baselines. If only a subset of experiments are reproducible, they should state which ones are omitted from the script and why.
        \item At submission time, to preserve anonymity, the authors should release anonymized versions (if applicable).
        \item Providing as much information as possible in supplemental material (appended to the paper) is recommended, but including URLs to data and code is permitted.
    \end{itemize}

\item {\bf Experimental setting/details}
    \item[] Question: Does the paper specify all the training and test details (e.g., data splits, hyperparameters, how they were chosen, type of optimizer) necessary to understand the results?
    \item[] Answer: \answerYes{}
    \item[] Justification: 
    The paper specifies the experimental settings needed to understand the simulation results. 
    The data-generating processes, treatment assignment mechanisms, error dependence structures, data splitting schemes, compared methods, evaluation metrics, and additional implementation details are described in Section 4 and Appendix C.
    \item[] Guidelines:
    \begin{itemize}
        \item The answer \answerNA{} means that the paper does not include experiments.
        \item The experimental setting should be presented in the core of the paper to a level of detail that is necessary to appreciate the results and make sense of them.
        \item The full details can be provided either with the code, in appendix, or as supplemental material.
    \end{itemize}

\item {\bf Experiment statistical significance}
    \item[] Question: Does the paper report error bars suitably and correctly defined or other appropriate information about the statistical significance of the experiments?
    \item[] Answer: \answerYes{}
    \item[] Justification: 
    The paper reports standard errors for the main simulation figures that support the empirical conclusions. 
    The standard errors are computed across Monte Carlo replications under each experimental condition, as described in Appendix C.
    \item[] Guidelines:
    \begin{itemize}
        \item The answer \answerNA{} means that the paper does not include experiments.
        \item The authors should answer \answerYes{} if the results are accompanied by error bars, confidence intervals, or statistical significance tests, at least for the experiments that support the main claims of the paper.
        \item The factors of variability that the error bars are capturing should be clearly stated (for example, train/test split, initialization, random drawing of some parameter, or overall run with given experimental conditions).
        \item The method for calculating the error bars should be explained (closed form formula, call to a library function, bootstrap, etc.)
        \item The assumptions made should be given (e.g., Normally distributed errors).
        \item It should be clear whether the error bar is the standard deviation or the standard error of the mean.
        \item It is OK to report 1-sigma error bars, but one should state it. The authors should preferably report a 2-sigma error bar than state that they have a 96\% CI, if the hypothesis of Normality of errors is not verified.
        \item For asymmetric distributions, the authors should be careful not to show in tables or figures symmetric error bars that would yield results that are out of range (e.g., negative error rates).
        \item If error bars are reported in tables or plots, the authors should explain in the text how they were calculated and reference the corresponding figures or tables in the text.
    \end{itemize}

\item {\bf Experiments compute resources}
    \item[] Question: For each experiment, does the paper provide sufficient information on the computer resources (type of compute workers, memory, time of execution) needed to reproduce the experiments?
    \item[] Answer: \answerNo{}
    \item[] Justification: 
    The paper does not provide detailed compute-resource information for each individual experiment. 
    The experiments are lightweight simulation studies used to illustrate the theoretical findings, and all reported experiments can be run on a standard laptop CPU without specialized hardware.
    \item[] Guidelines:
    \begin{itemize}
        \item The answer \answerNA{} means that the paper does not include experiments.
        \item The paper should indicate the type of compute workers CPU or GPU, internal cluster, or cloud provider, including relevant memory and storage.
        \item The paper should provide the amount of compute required for each of the individual experimental runs as well as estimate the total compute. 
        \item The paper should disclose whether the full research project required more compute than the experiments reported in the paper (e.g., preliminary or failed experiments that didn't make it into the paper). 
    \end{itemize}
    
\item {\bf Code of ethics}
    \item[] Question: Does the research conducted in the paper conform, in every respect, with the NeurIPS Code of Ethics \url{https://neurips.cc/public/EthicsGuidelines}?
    \item[] Answer: \answerYes{}
    \item[] Justification: 
    The research conforms to the NeurIPS Code of Ethics. 
    The paper is theoretical and simulation-based, does not involve human subjects, private or sensitive data, or high-risk deployment, and the submission is prepared to preserve anonymity during review.
    \item[] Guidelines:
    \begin{itemize}
        \item The answer \answerNA{} means that the authors have not reviewed the NeurIPS Code of Ethics.
        \item If the authors answer \answerNo, they should explain the special circumstances that require a deviation from the Code of Ethics.
        \item The authors should make sure to preserve anonymity (e.g., if there is a special consideration due to laws or regulations in their jurisdiction).
    \end{itemize}

\item {\bf Broader impacts}
    \item[] Question: Does the paper discuss both potential positive societal impacts and negative societal impacts of the work performed?
    \item[] Answer: \answerYes{}
    \item[] Justification: 
    The paper discusses the implications of the results for uncertainty quantification in individualized decision-making settings. 
    A positive societal impact is that the results may prevent overconfident use of distribution-free ITE prediction intervals under insufficient assumptions; a potential negative impact is that the results could be misinterpreted as discouraging all uncertainty quantification for causal effects, which we mitigate by clarifying the scope of the impossibility results and by presenting a positive result under additional structure.
    \item[] Guidelines:
    \begin{itemize}
        \item The answer \answerNA{} means that there is no societal impact of the work performed.
        \item If the authors answer \answerNA{} or \answerNo, they should explain why their work has no societal impact or why the paper does not address societal impact.
        \item Examples of negative societal impacts include potential malicious or unintended uses (e.g., disinformation, generating fake profiles, surveillance), fairness considerations (e.g., deployment of technologies that could make decisions that unfairly impact specific groups), privacy considerations, and security considerations.
        \item The conference expects that many papers will be foundational research and not tied to particular applications, let alone deployments. However, if there is a direct path to any negative applications, the authors should point it out. For example, it is legitimate to point out that an improvement in the quality of generative models could be used to generate Deepfakes for disinformation. On the other hand, it is not needed to point out that a generic algorithm for optimizing neural networks could enable people to train models that generate Deepfakes faster.
        \item The authors should consider possible harms that could arise when the technology is being used as intended and functioning correctly, harms that could arise when the technology is being used as intended but gives incorrect results, and harms following from (intentional or unintentional) misuse of the technology.
        \item If there are negative societal impacts, the authors could also discuss possible mitigation strategies (e.g., gated release of models, providing defenses in addition to attacks, mechanisms for monitoring misuse, mechanisms to monitor how a system learns from feedback over time, improving the efficiency and accessibility of ML).
    \end{itemize}
    
\item {\bf Safeguards}
    \item[] Question: Does the paper describe safeguards that have been put in place for responsible release of data or models that have a high risk for misuse (e.g., pre-trained language models, image generators, or scraped datasets)?
    \item[] Answer: \answerNA{}
    \item[] Justification: 
    The paper does not release high-risk models, scraped datasets, or sensitive data that would require safeguards for responsible release. 
    The released materials, if any, consist only of anonymized simulation code for reproducing synthetic experiments.
    \item[] Guidelines:
    \begin{itemize}
        \item The answer \answerNA{} means that the paper poses no such risks.
        \item Released models that have a high risk for misuse or dual-use should be released with necessary safeguards to allow for controlled use of the model, for example by requiring that users adhere to usage guidelines or restrictions to access the model or implementing safety filters. 
        \item Datasets that have been scraped from the Internet could pose safety risks. The authors should describe how they avoided releasing unsafe images.
        \item We recognize that providing effective safeguards is challenging, and many papers do not require this, but we encourage authors to take this into account and make a best faith effort.
    \end{itemize}

\item {\bf Licenses for existing assets}
    \item[] Question: Are the creators or original owners of assets (e.g., code, data, models), used in the paper, properly credited and are the license and terms of use explicitly mentioned and properly respected?
    \item[] Answer: \answerYes{}
    \item[] Justification:
    The paper uses existing open-source implementations of some baseline methods, and the original authors, repositories, versions, and licenses are credited in Appendix C. All such assets are used in accordance with their stated license terms.
    \item[] Guidelines:
    \begin{itemize}
        \item The answer \answerNA{} means that the paper does not use existing assets.
        \item The authors should cite the original paper that produced the code package or dataset.
        \item The authors should state which version of the asset is used and, if possible, include a URL.
        \item The name of the license (e.g., CC-BY 4.0) should be included for each asset.
        \item For scraped data from a particular source (e.g., website), the copyright and terms of service of that source should be provided.
        \item If assets are released, the license, copyright information, and terms of use in the package should be provided. For popular datasets, \url{paperswithcode.com/datasets} has curated licenses for some datasets. Their licensing guide can help determine the license of a dataset.
        \item For existing datasets that are re-packaged, both the original license and the license of the derived asset (if it has changed) should be provided.
        \item If this information is not available online, the authors are encouraged to reach out to the asset's creators.
    \end{itemize}

\item {\bf New assets}
    \item[] Question: Are new assets introduced in the paper well documented and is the documentation provided alongside the assets?
    \item[] Answer: \answerYes{}
    \item[] Justification: 
    The paper releases anonymized simulation code as supplementary material to reproduce the main experimental results. 
    The code package includes documentation and instructions for generating the synthetic data and reproducing the reported figures, and no new dataset, pretrained model, or human-subject asset is introduced.
    \item[] Guidelines:
    \begin{itemize}
        \item The answer \answerNA{} means that the paper does not release new assets.
        \item Researchers should communicate the details of the dataset\slash code\slash model as part of their submissions via structured templates. This includes details about training, license, limitations, etc. 
        \item The paper should discuss whether and how consent was obtained from people whose asset is used.
        \item At submission time, remember to anonymize your assets (if applicable). You can either create an anonymized URL or include an anonymized zip file.
    \end{itemize}

\item {\bf Crowdsourcing and research with human subjects}
    \item[] Question: For crowdsourcing experiments and research with human subjects, does the paper include the full text of instructions given to participants and screenshots, if applicable, as well as details about compensation (if any)? 
    \item[] Answer: \answerNA{}
    \item[] Justification: 
    The paper does not involve crowdsourcing, user studies, or research with human subjects. 
    All experiments are based on synthetic data generated from the simulation designs described in Section 4 and Appendix C.
    \item[] Guidelines:
    \begin{itemize}
        \item The answer \answerNA{} means that the paper does not involve crowdsourcing nor research with human subjects.
        \item Including this information in the supplemental material is fine, but if the main contribution of the paper involves human subjects, then as much detail as possible should be included in the main paper. 
        \item According to the NeurIPS Code of Ethics, workers involved in data collection, curation, or other labor should be paid at least the minimum wage in the country of the data collector. 
    \end{itemize}

\item {\bf Institutional review board (IRB) approvals or equivalent for research with human subjects}
    \item[] Question: Does the paper describe potential risks incurred by study participants, whether such risks were disclosed to the subjects, and whether Institutional Review Board (IRB) approvals (or an equivalent approval/review based on the requirements of your country or institution) were obtained?
    \item[] Answer: \answerNA{}
    \item[] Justification: 
    The paper does not involve crowdsourcing, user studies, or research with human subjects. 
    Therefore, IRB approval or equivalent human-subjects review is not applicable.
    \item[] Guidelines:
    \begin{itemize}
        \item The answer \answerNA{} means that the paper does not involve crowdsourcing nor research with human subjects.
        \item Depending on the country in which research is conducted, IRB approval (or equivalent) may be required for any human subjects research. If you obtained IRB approval, you should clearly state this in the paper. 
        \item We recognize that the procedures for this may vary significantly between institutions and locations, and we expect authors to adhere to the NeurIPS Code of Ethics and the guidelines for their institution. 
        \item For initial submissions, do not include any information that would break anonymity (if applicable), such as the institution conducting the review.
    \end{itemize}

\item {\bf Declaration of LLM usage}
    \item[] Question: Does the paper describe the usage of LLMs if it is an important, original, or non-standard component of the core methods in this research? Note that if the LLM is used only for writing, editing, or formatting purposes and does \emph{not} impact the core methodology, scientific rigor, or originality of the research, declaration is not required.

    \item[] Answer: \answerNA{}
    \item[] Justification: 
    The core method development, theoretical results, proofs, and experiments in this paper do not involve LLMs as an important, original, or non-standard component. 
    Any use of LLMs, if applicable, was limited to writing, editing, or formatting assistance and did not affect the methodology, scientific rigor, or originality of the research.
    \item[] Guidelines:
    \begin{itemize}
        \item The answer \answerNA{} means that the core method development in this research does not involve LLMs as any important, original, or non-standard components.
        \item Please refer to our LLM policy in the NeurIPS handbook for what should or should not be described.
    \end{itemize}

\end{enumerate}